\documentclass[10pt,twocolumn,twoside]{IEEEtran} 

\usepackage{amsmath,amssymb,graphicx,cite,color}
\usepackage{generic}
\usepackage{cite}
\usepackage{amsmath,amsfonts}

\usepackage{graphicx}
\usepackage{algorithm,algorithmic}
\usepackage{hyperref}
\hypersetup{hidelinks}
\usepackage{textcomp}
\usepackage{comment}
\usepackage{gensymb}
\usepackage[dvipsnames]{xcolor} 

\usepackage{enumitem}
\usepackage{amssymb, mathbbol, mathtools, amsthm} 
\usepackage{graphicx}
\usepackage{csquotes} 
\usepackage{hyperref} 
\usepackage{textcomp}
\usepackage[export]{adjustbox}
\usepackage{subcaption}
\usepackage{setspace}
\usepackage{float}
\usepackage{array}
\usepackage{graphics} 
\usepackage{epsfig} 
\usepackage{caption}
\usepackage{mathptmx} 
\usepackage{times}
\usepackage{tabularx,multirow}
\newtheorem{theorem}{Theorem}
\newtheorem{lemma}{Lemma}
\newtheorem{remark}{Remark}
\newtheorem{proposition}{Proposition}
\theoremstyle{definition}
\newtheorem{definition}{Definition}
\newtheorem{example}{Example}
\theoremstyle{corollary}
\theoremstyle{assumption}

\newtheorem{corollary}{Corollary}
\def\BibTeX{{\rm B\kern-.05em{\sc i\kern-.025em b}\kern-.08em
    T\kern-.1667em\lower.7ex\hbox{E}\kern-.125emX}}
\begin{document}
\title{Laplacian Flows in Complex-valued Directed Networks: Analysis, Design, and Consensus}
\author{Aditi Saxena, \IEEEmembership{Student Member, IEEE}, Twinkle Tripathy, \IEEEmembership{Senior Member, IEEE} and Rajasekhar Anguluri, \IEEEmembership{Member, IEEE}
\thanks{$^{1}$Aditi Saxena is a research scholar with the Department of Electrical Engineering, Indian Institute of Technology Kanpur, India. (email: {\tt\small saditi23@iitk.ac.in)}}%
\thanks{$^{2}$Twinkle Tripathy is with the Department of Electrical Engineering, Indian Institute of Technology Kanpur, India. (email: 
{\tt\small ttripathy@iitk.ac.in)}}%
\thanks{$^{3}$Rajasekhar Anguluri is with the Department of Computer Science and Electrical Engineering,  University of Maryland, Baltimore County, MD 85281, USA. (e-mail: {\tt\small rajangul@umbc.edu})}
 }
\maketitle
\begin{abstract}
In the interdisciplinary field of network science, a complex-valued network, with edges assigned complex weights, provides a more nuanced representation of relationships by capturing both the magnitude and phase of interactions. Additionally, an important application of this setting arises in distribution power grids. Motivated by the richer framework, we study the necessary and sufficient conditions for achieving consensus in both strongly and weakly connected digraphs. The paper establishes that complex-valued Laplacian flows converge to consensus subject to an additional constraint termed as `real dominance' which relies on the phase angles of the edge weights. 
Our approach builds on the complex Perron-Frobenius properties to study the spectral properties of the Laplacian and its relation to graphical conditions. Finally, we propose modified flows that guarantee consensus even if the original network does not converge to consensus. Additionally, we explore diffusion in complex-valued networks as a dual process of consensus and simulate our results on synthetic and real-world networks.
\end{abstract}
 \section{Introduction}\label{sec:introduction}
A Laplacian flow describes the state evolution of a vector-valued continuous linear time-invariant homogeneous system, where the state matrix is the negative of a Laplacian associated with an underlying finite-vertex network. The choice of Laplacian is governed by applications including the human brain network; power, water, and gas networks; societal groups; robotics; and transportation. In the literature, important results have been established on network flows based on network and matrix-theoretic conditions to achieve properties, including controllability, observability, asymptotic convergence, and importantly \textit{consensus}. Nonetheless, the edge-weights in a network are real-valued in many studies. We depart from this custom by studying networks with complex-valued edges; hereafter, \emph{complex-valued networks}.

The central mathematical object in the study of flows is the map $e^{-Lt}$, where $t \geq 0$ and $L \in \mathbb{C}^{n \times n}$ is the Laplacian (Section \ref{sec:preliminaries}). For real-valued, non-negative Laplacians, many properties of this map follow from standard results, including the Perron-Frobenius theory for dominant eigenvalues and the positive semi-definiteness of $L$. One important property is the eventual exponential positivity (or non-negativity) of $e^{-Lt}$ (see \cite{fontan2022multiagent,fontan2021properties}) that allows to study the steady-state behavior of flows using the network reachability and eigen-spectrum of $L\in \mathbb{R}^{n\times n}$.

The desirable properties of real-valued digraphs do not carry over in a straightforward manner to complex-valued Laplacian networks. For example, complex-symmetric Laplacians $L = L^\top \in \mathbb{C}^{n \times n}$ may have eigenvalues that are not real or positive (see Ex.~\eqref{ex:unstable eigenvalues}), and $e^{-Lt}$ may not exhibit eventual exponential behavior. Thus, classical reachability fails to capture the null-space of $L$ corresponding to its zero eigenvalue, unlike in the real-valued setting. These intriguing properties are intrinsic to complex-valued dynamics and are absent in networks with non-negative real weights. The goal of this work is to formalize these properties and develop suitable generalizations of eventual positivity to the complex domain, enabling analogs of known results for real-valued Laplacians (see~\cite{olfati2007consensus,fontan2022multiagent}).

\smallskip 

\vspace{0.3mm}


\textit{Motivation}: 
Our interest in complex-valued networks is due to two applications: (i) distribution power grids and (ii) opinion dynamics in social networks. Operations in power grids, such as static state estimation, stability, and control, depend on the complex-valued bus admittance matrix. Traditionally, only its imaginary part (which is real-valued) has been used to simplify analysis. However, recent research has highlighted the need to consider the complex-valued admittance matrix in distribution networks for network reconstruction and power-flow problems \cite{Stevenlow}. It is worth acknowledging the extensive work on the mathematical theory of multi-port electrical circuits and network synthesis from the past century \cite{anderson2013network, seshu1961reed, wang2024first}, where the admittance matrix can have complex values or polynomials. 

In social networks, specifically concerning opinion dynamics, complex-valued edges allow agents (nodes) to be described beyond simple binary relations like friendship or foes. These relationships can be achieved by expressing edge-weights in the polar form (\( r e^{\iota \beta} \)), where the phase angle \( \beta \in (-90^{\circ},90^{\circ})\) captures the nature and strength of the relationship \cite{tian2024structuralbalance}. Finally, complex-valued networks are employed to study frustration in digraphs \cite{gong2021directed}. 
Other areas where complex-valued networks play a key role include quantum computation \cite{kubota2021quantum, mukai2020discrete, mn2015continuous, godsil2023discrete};  
 electrical circuits \cite{muranova2020electrical}; complex-valued neural networks\cite{kobayashi2010nnl}; and communication networks \cite{tong2019mimo}. In image processing, complex values provide an extra rotational information of the visual data \cite{koteswar2023fccns}. In signal networks, complex-valued edge-weights encode magnitude and phase transfer characteristics of interferometer networks \cite{krawciw2024small}. 
 Finally, see \cite{bottcher2024complex} for a survey covering various aspects of complex-valued networks for network science.
\smallskip 

 \textit{Literature Review}:
Consensus problems have been the mainstay in the literature of multi-agent networks \cite{li2017consensus,olfati2007consensus,zhang2020consensus,pham2019consensus}. Applications range from quantum, power, and brain networks to coordination in autonomous vehicles \cite{olfati2007consensus,garindistributed,fontan2021properties,dorfler2012synchronization}. At the heart of consensus problem is Laplacian flows. 
Ref.~\cite{olfati2007consensus} and the subsequent work set the stage for studying Laplacian flows by highlighting that non-negative matrices (e.g., M-matrix; Perron and stochastic matrices) arise in modeling both large-scale engineering and communication systems. 
Further, \cite{fontan2022multiagent} provided normality as a sufficient condition for achieving consensus in signed digraphs using the Perron-Frobenius properties of real-valued matrices. 

{Studies on complex-valued networks for analyzing network dynamics, including consensus, is limited within the broader field of network science and virtually absent in control theory. A few exceptions exist. Ref.\cite{reff2012spectral} studies Hermitian Laplacians associated with complex unit gain graphs and derives spectral bounds. Ref.\cite{tian2024structuralbalance} extended structural notions such as balance, antibalance, and unbalance to the Hermitian setting, along with random walk dynamics and a spectral clustering method. In Ref.\cite{lin2013leader}, complex edge weights were used in a leader–follower framework to model planar formations via rotated lines of sight. Refs.\cite{dong2014complex} and  \cite{DONG20161} explored singularity conditions for non-Hermitian Laplacians using a generalized notion of balance, where a cycle is positive if the product of its edge weights is real and positive~\cite{dong2015consensus}.

Most existing studies focus on undirected or Hermitian networks. Here, we relax those assumptions and consider general non-Hermitian digraphs. While prior works have addressed such cases, they often rely on balance assumptions. In contrast, our analysis applies to both balanced and unbalanced complex-valued networks. This enables us to investigate broader questions: Can non-trivial consensus emerge in unbalanced complex digraphs? Do the conditions align with the real-valued case? And if not, what accounts for the differences?} 

\subsection{Contributions}
This research builds on our preliminary studies on complex-valued Laplacian and pseudoinverse flows \cite{saxena2024realeventualexponentialpositivity,saxena2024flowscomplexvaluedlaplacianspseudoinverses} in undirected and weight-balanced digraphs\footnote{ A digraph is weight-balanced if each node’s in-degree equals its out-degree.}, having non-Hermitian Laplacians. For these digraphs, the eigenvector corresponding to the zero eigenvalue being both real and positive analogous to the case of the real-valued Laplacian, allowed us to develop sufficient conditions for consensus. We build upon this prior work to accommodate non weight-balanced digraphs. Moreover, we relax the signed edge weight conditions and allow our imaginary part of edge-weights to be positive or negative. We summarize our main contributions: 
\begin{enumerate}[label=(\roman*)]
 \item Refs. \cite{fontan2022multiagent,fontan2021properties} masterfully highlighted the role of eventually exponential property of $L \in \mathbb{R}^{n \times n}$ for flows. Building on this line of work, we introduce real eventually exponentially positive (rEEP) and negative property (rEENN) for complex-valued matrices, not necessarily Laplacians. We characterize these properties in terms of the dominant eigenvalues and the associated eigenvectors, a complex-valued analogue to the Perron-Frobenius theory.
\item For the flow system described in Eqn.~\eqref{eq:Lflows}, we demonstrate that the negated complex-valued Laplacian is rEEP for strongly connected digraphs and rEENN for weakly connected digraphs. These results depend on easily verifiable conditions on the angles of the edge-weights (Eqn.~\eqref{eq:tan}) and the (real) dominance of eigenvectors corresponding to the zero eigenvalue of $-L\in \mathbb{C}^{n\times n}$.

\item We give necessary and sufficient conditions to achieve consensus in strongly and weakly connected digraphs that satisfy the phase angle and dominance conditions. For digraphs that do not satisfy these conditions, we propose an edge modification method which yields in networks that achieve consensus. Our method is in the spirit of the pole-placement approach for achieving asymptotic stability in the linear dynamical systems. 



\end{enumerate}
To the best of our knowledge, this work presents the first systematic analysis of complex-valued flows, extending core concepts and tools that have, until now, been restricted to real-valued flows. We present numerous examples throughout the paper to both validate our theoretical results and to: (i) reveal the counterintuitive behavior of complex-valued matrices, and (ii) contrast these findings with their real-valued counterparts.
\section{Preliminaries and Problem Formulation}
\label{sec:preliminaries}
Let $\mathbb{C}^{n \times n}$ denote the space of $n\times n$ complex-valued matrices. For $n=1$, the space is $\mathbb{C}$. Let ${\Re}(m) \text{ and } {\Im}(m)$ denote the real and imaginary parts of $m \in \mathbb{C}$, respectively. The imaginary unit is denoted by $\iota=\sqrt{-1}$.  A diagonal matrix is denoted by $\operatorname{diag}(d_{11},..,d_{nn})$, where $d_{ii}\in \mathbb{C}$. The right half of the complex- and real-planes  are $\mathbb{C}_{+} = \{ {\Re}(m) \geq 0, \text{ for all } m \in \mathbb{C}\}$ and $\mathbb{R}_{+} = \{ m \geq 0, \text{ for all } m \in \mathbb{R}\}$, respectively, where the open right-half of complex plane is denoted as ORHP. The $n$-dimensional vectors $\mathbb{1}_n$ and $\mathbb{0}_n$ are all ones and zeros vectors.

\smallskip 
\textit{Matrix Algebra}: A matrix $M \in \mathbb{C}^{n \times n}$  is said to be real non-negative if $\Re(m_{ij})\geq0$ for all the entries 
(hereafter $\Re(M)\geq 0$). A matrix is real eventually non-negative (rENN hereafter) if $\Re(M^{k})\geq 0 $ for some $ k\geq k_{0}$. Real positive and real eventually positive matrices (rEP hereafter) satisfy the above inequalities with strict inequality. $M^{H}$ and $M^{\top}$ denote the complex conjugate transpose and transpose of $M$. 
    The eigenspectrum of $M$ is $\operatorname{spec}(M)=\left \{ \lambda_{1},\lambda_{2},...,\lambda_{n} \right \}, \text{ where } \lambda_{i}$ is the $i^{th}$ eigenvalue of $M$. The Jordan decomposition is given by $M = VJW^H$, where $V$ and $W$ are the matrices of right (column) eigenvectors and of left (row) eigenvectors, respectively. Without loss of generality we let $\lambda_1$ to be dominant: $|\lambda_1|\geq |\lambda_j|$, for $j\ne 1$. Denote $v\in V$ and $w \in W $ to be the dominant right and left eigenvectors corresponding to $\lambda_1$; and these vectors are normalized, $w^H v=1$. The spectral abscissa of $M$ is the maximum of the real parts of all the eigenvalues and is denoted as $\lambda_s$. A complex matrix is marginally stable if $\Re(\lambda_i(M)) \leq 0$ and has corank $d$ if the dimension of the kernel space of $M$ is $d$.  
A $\lambda_{i}$ is semi-simple if algebraic multiplicity equals geometric multiplicity and is simple if both the multiplicities are one. A permutation matrix $P$ is a square binary matrix having a single entry one in each row and column, and all other entries are zeros \cite{vargamatrix}. 

We recall a few standard results on Perron-Frobenius (hereafter PF) properties for complex-valued matrices that play a crucial role in our analysis. More details are in \cite{varga2012,rugh2010cones}. 
\begin{definition}\label{def1:PF}(\textbf{\textit{Complex PF Property\cite{varga2012}}})
 A matrix $M\in\mathbb{C}^{n \times n}$ has the complex Perron-Frobenius property if its dominant eigenvalue $\lambda_{1}$ is real and positive. The associated (non-zero) right eigenvector ${v}\in \mathbb{C}^n$
is such that ${\Re}({v}) \geq \mathbb{0}_n$.
\end{definition}
\begin{definition}\label{def2:strong PF}(\textbf{\textit{Strong Complex PF Property\cite{varga2012}}}) A matrix $M \in \mathbb{C}^{n\times n}$ has the strong complex Perron-Frobenius property if its dominant eigenvalue $\lambda_{1}$ is real, positive and simple with $\lambda_{1}>\left|\lambda_{i}\right|$ for all $i \in \left \{2,...,n \right \}$. The right eigenvector ${v}\in \mathbb{C}^n$ associated with $\lambda_1$ is such that  $\Re({v})>\mathbb{0}_n$.
\end{definition}
\begin{lemma}\label{lemma:PF and cone}(\textbf{\textit{Cone contraction property\cite{rugh2010cones}}})
A bounded linear operator\footnote{A bounded linear operator $T:V \longrightarrow W$ on two normed linear spaces $V$ and $W$ is such that $|| Tx|| \leq c ||x||$ for all $x \in V$ for some $c>0$.} $T$ has a simple and positive dominant eigenvalue if and only if it is a strict contraction of a regular $\mathbb{C}-$cone.\footnote{A subset $\mathcal{C}$ is a closed complex cone if it is closed and $\mathbb{C}$-invariant (i.e.\ $\mathcal{C} = \mathbb{C} \mathcal{C}$) and $\mathcal{C} \neq \{0\}$. Further, a closed complex cone $\mathcal{C}$ is \emph{proper} if it contains no complex planes, $i.e.,$ if $x$ and $y$ are independent vectors, then $\operatorname{span}\{x, y\} \nsubseteq \mathcal{C}$.
\textcolor{black}{A proper closed complex cone is denoted as a $\mathbb{C}-$cone. A regular cone is a cone that is both inner regular (has nonempty interior) and outer regular (has bounded width)
 A complex linear map $T:E_1 \longrightarrow E_2$ is called a cone contraction if it reduces the cone distance between points in the mapped vector space. Interested readers may refer \cite{rugh2010cones} for details.}}  
\end{lemma}
\begin{definition}\label{def:real EEP}(\textbf{\textit{Real Eventually Exponentially Positive matrix (rEEP)\cite{saxena2024realeventualexponentialpositivity}}}) A matrix $M \in \mathbb{C}^{n\times n}$ is 
rEEP if there exists an exponential index $t_{0}$ such that $\Re(e^{M t})>0$ for all time $t \geq t_{0}$.
 \end{definition}
 \begin{definition}\label{def:real EENN}(\textbf{\textit{Real Eventually Exponentially Non-Negative matrix (rEENN)}}) A matrix $M \in \mathbb{C}^{n\times n}$ is 
 rEENN if there exists $t_{0}$ such that $\Re(e^{M t})\geq 0$ for all time $t \geq t_{0}$.
 \end{definition}
 \begin{definition}\label{def5}(\textbf{\textit{Irreducible Matrix \cite{vargamatrix}}})
For $n\geq 2$, an irreducible matrix $M\in \mathbb{C}^{n\times n}$ is such that there does not exist an $n\times  n$ permutation matrix $P$ such that 
\begin{equation}\label{eq:irreducible form}
 PMP^{T}=\begin{bmatrix}
M_{11} & M_{12}\\ 
 0& M_{22} 
\end{bmatrix},
\end{equation}
where $M_{11}$ is an $r\times r$ submatrix and $M_{22}$ is an $(n-r)\times(n-r)$ submatrix where $1\leq r<n$. 
\end{definition}

\textit{Graph Theory}:  
We represent a network by the tuple $\mathcal{G}(A)=(\mathcal{V},\mathcal{E},A)$, where $\mathcal{V}$ and $\mathcal{E}$ are the sets of vertices and edges, respectively, and  $A \in  \mathbb{C}^{n \times n}$ is the adjacency matrix. If $A=A^\top$ (i.e., complex symmetric), then $\mathcal{G}(A)$ is undirected, otherwise it is directed (or digraph).
A digraph is weakly connected if its undirected version is connected and it is strongly connected if there exists a directed path between any pair of vertices. In digraphs, out-degree of a vertex is the sum of edge-weights of all the outgoing edges, whereas in-degree is the sum of edge-weights of incoming edges. Collectively for all nodes, the degree matrices are given by $D_{out}=\operatorname{diag}(A\mathbb{1}_{n})$ and $D_{in}=\operatorname{diag}(A^{T}\mathbb{1}_{n})$. 
 A node is globally reachable if all the other nodes have a directed path to that node. A node with zero out-degree is sink whereas the source node has zero in-degree \cite{bullo2018lectures}.

There are several ways to define Laplacian matrices for complex-valued networks. In \cite{tian2024structuralbalance,dong2015consensus}, the degree of a node equals the sum of the magnitudes of edge-weights and the Laplacian is $L=D_{out}-A$, where $A$ is the complex adjacency matrix and the diagonal $D_{out}$ consists of the magnitudes of degree nodes. We define our Laplacian in a similar way but the degree matrix is complex-valued rather than considering the magnitude of edge-weights. The definition is motivated by physical and power networks where the diagonal entries of the Laplacian matrix are complex-valued \cite{Networksgraphtheory2018}.  

\subsection{Problem Formulation and Main Result}
Consider a complex-valued digraph $\mathcal{G}(A)$ with $n$ nodes with complex-valued edge-weights. In our analysis, we make use of polar as well as rectangular form representation of the complex number. Thus for $A=(a_{ij}) \in \mathbb{C}^{n\times n}$, we let 
$a_{ij}=r_{ij} \angle \beta_{ij}$, where $r\geq 0$ is the magnitude and $\beta_{ij} \in (-90^{\circ},90^{\circ})$ ({\color{black} measured in degrees}) is the phase angle.  
The phase angles quantify the strength of interactions among the nodal agents. For example, when $\beta_{ij}=0$ we have the unsigned real-valued network. In the rectangular form $a_{ij}=z_\text{real}+z_\text{imag}\iota $.

Let $x(t) \in \mathbb{C}^n$ be the state vector and consider the (generalized) Laplacian flow system
\begin{equation}\label{eq:Lflows}
    \dot{x}(t)=-\bar{L}x(t), \quad t \geq0, 
\end{equation}
where $\bar{L}:\mathbb{C}^{n\times n}\to \mathbb{C}^{n\times n}$. We primarily consider $\bar{L}=L,$ and $\bar{L}=L_m,$ (Laplacian of a modified network) and focus on the steady state behavior of the state trajectories of Eqn. \eqref{eq:Lflows}. A comprehensive study concerning the pseudoinverse Laplacian flows ($\bar{L}=L^{\dagger}$) is provided in \cite{saxena2024flowscomplexvaluedlaplacianspseudoinverses}.

Our primary result establishes the necessary and sufficient conditions for achieving consensus in complex-valued strongly connected and weakly connected digraphs. Precisely, consensus ensures all components of the state $x(t)$ to be nonzero and identical as $t \to \infty$. The final consensus value may take arbitrary $z \in \mathbb{C}$ depending on the initial conditions. 

The existing results for (real-valued) signed and unsigned networks are readily inapplicable for complex-valued digraphs. As such we develop the required machinery based on (i) spectral properties of the (complex-valued) Laplacian matrices, (ii) complex Perron-Frobenius theory, and (iii) the stabilization theorem in Lemma \ref{lemma:existence of D}. 
\section{Eventual Exponential Properties of Complex-Valued Matrices}
\label{sec:complex matrices} 
We begin studying the spectral properties of complex-valued matrices by establishing relationships between the Perron-Frobenius properties in Defs. \ref{def1:PF} and \ref{def2:strong PF} and the rEEP and rEENN notions in Defs. \ref{def:real EEP} and \ref{def:real EENN}. Our approach relies on the spectral decomposition of complex matrices and is inspired from \cite{varga2012}. These relationships serve as a starting point for analyzing (generalized) Laplacian flows in strongly and weakly connected digraphs.
Let $\lambda_1$ be the dominant eigenvalue of $M \in\mathbb{C}^{n\times n}$ and $v \in \mathbb{C}^n$ be the corresponding dominant right eigenvector. Define
\begin{subequations}
\begin{align}
\label{eq:PFmatrix}
\hspace{-2mm} \mathcal{P} &\!=\! \{M \in\mathbb{C}^{n\times n}: \lambda_1(M) \in \mathbb{R}_{+}  \text{ is simple}\},\\
\label{eq:strongPFmatrix}
\hspace{-2mm} \mathcal{PF} &\!=\! \{M \in\mathbb{C}^{n\times n}:~ \lambda_1(M) \in \mathbb{R}_{+}\text{ is simple and } \Re{(v)}\!>\!0\}.
\end{align}
\end{subequations}
 The inequality  $\Re{(v)}>0$ in Eqn.~\eqref{eq:strongPFmatrix} is element-wise. The set $\mathcal{P}$ in Eqn.~\eqref{eq:PFmatrix} is a subset of matrices satisfying the complex PF property in Def.~\ref{def1:PF}. Finally, $\mathcal{PF}$ in Eqn.~\eqref{eq:strongPFmatrix} is the set of matrices satisfying the strong complex PF property in Def.~\ref{def2:strong PF}.

 {\color{black}We recall the classes of rEP or rENN matrices: the former class satisfies $\Re(M^{k})> 0 $ and the latter class needs $\Re(M^{k})\geq 0$, for some $ k\geq k_{0}$. These matrices play a key role in analyzing system theoretic concepts such as reachability and holdability (see \cite{friedland1978inverse} for details). In order to establish rEP or rEENN properties of matrices, we introduce the notion of \emph{real dominance} of eigenvectors.
\begin{definition}\label{def:real dominance}(\textbf{\textit{Real dominance}}) A vector $z\in \mathbb{C}^n$ is \emph{real dominant} if $\Re({z})\geq\left | \Im({z}) \right|$. Equivalently, the phase angle of each entry of $z$ lies in the range $[-45^{\circ},45^{\circ}]$.
\end{definition}
When the inequality holds strictly, i.e., $\Re(z) > |\Im(z)|$, the vector is referred to as \emph{strictly real dominant}.
Next, we develop the conditions to relate strong complex PF property (codified in $\mathcal{PF}$) with rEP and rEEP (Def. \ref{def:real EEP}).

\begin{theorem}\label{thm:stronglemma}
Let $M \in \mathbb{C}^{n \times n}$. Suppose $\lambda_s$ is the spectral abscissa of $M$; $v$ and $w$ are the corresponding dominant right and left eigenvectors, respectively, which satisfy real dominance (Def. \ref{def:real dominance}) and the normalisation condition $w^Hv=1$. Then, the following statements are equivalent for some scalar $d\geq 0$:
\begin{enumerate}[label=(\roman*)]
\item $ (M+dI) \text{ and } (M+dI)^{H} \in \mathcal{PF}$.
\item $(M+dI)$ is rEP.
\item $M$ is rEEP.
\end{enumerate}  
\end{theorem}
\begin{proof}  
$(i) \implies (ii)$: Consider the Jordan decomposition of $M+dI=V J W^H$, where $J$ comprises Jordan blocks of all the eigenvalues. For a suitable choice of $d\geqslant 0$, $(\lambda_s+d)$ can be the dominant eigenvalue of $M+dI$. Then partition $M+dI$ as
$$
M+dI=\left[v \mid V_{n, n-1}\right]\left[\begin{array}{l|l}
(\lambda_{s}+d) & 0 \\
\hline 0 & J_{n-1}
\end{array}\right]\left[\begin{array}{l}
w^H\\
\hline W_{n-1, n}^H
\end{array}\right].
$$
{\color{black} Since}  $(M+dI) \text{ and } (M+dI)^{H} \in \mathcal{PF}, \text{ we have } \Re(v)>0$ and $\Re(w)>0$ from Eqn. \eqref{eq:strongPFmatrix}. Thus, for some $k\geq 0$, we have
\begin{equation}\label{eq: i implies ii}
\frac{(M+dI)^{k}}{\left(\lambda_{s}+d\right)^{k}}=\left[v \mid V_{n-1, n}\right]\left[\begin{array}{c|c}
1 & 0 \\
\hline 0 & \frac{(J_{n-1})^{k}}{\left(\lambda_{s}+d\right)^{k}}
\end{array}\right]\left[\begin{array}{c}
w^H\\
\hline W_{n-1, n}^H
\end{array}\right],
\end{equation}
where $J_{n-1}=\operatorname{diag}(\lambda_2+d,\ldots,\lambda_n+d)$ 
and $\lambda_{i} \in \operatorname{spec}(M)$.
Taking the limit $k\to \infty$ on the either sides of Eqn. \eqref{eq: i implies ii}, we get
\begin{align*}
\lim_{k \rightarrow \infty} \frac{(M+d I)^{k}}{\left(\lambda_{s}+d\right)^{k}}=v w^H.
\end{align*} 
Owing to the real dominance of $v$ and $w$, equating the real parts yields 
\begin{equation}\label{eq:Matsteadystate}
\lim _{k \rightarrow \infty} \Re  \left ( \frac{(M+dI)^{k}}{(\lambda_{s}+d)^{k}}\right )=\Re({v}{w}^H) > \mathbb{0}_{n\times n}.    
\end{equation}
Eqn. \eqref{eq:Matsteadystate} implies the existence of a positive integer $ k_0>0$ such that $\Re  \left ((M+dI)^{k}\right)>0$ for all $k \geq k_0$; meaning $(M+dI)$ is rEP.

 $(ii)\implies(i):$ Assume that $\Re{((M+dI)^k)}>0 \text{ for all }k \geq k_0$ and define $B=(M+dI)$. The scalar $d\geq 0$ can be suitably chosen such that $B$ becomes a non-nilpotent matrix. Then, it has been shown in \cite[Thm. 2.3]{varga2012} that $B,B^{H}\in\mathcal{PF}$. 
 
$(ii)\implies(iii):$  We need to show $M$ is rEEP, i.e., $\Re (e^{M t})>0$. Consider $e^{Mt}=e^{-dt}e^{(M+dI)t}$, where the scalar $e^{-dt}>0$. Thus, it suffices to prove that there exists $t_0>0$ such that $\Re(e^{(M+dI)t})>0$ for all $t \geq t_0$. 

We know that $e^{(M+dI)t} = \sum_{k=0}^{\infty}\frac{(M+dI)^{k} t^{k}}{k!}$. As per our hypothesis, $(M+dI)$ is rEP; so there exists $k_0\in \mathbb{N}$ such that $\Re\left((M+d I)^{k}\right)> \mathbb{0}_{n \times n}$ for every $k\geq k_0$. We can then re-write the matrix exponential as
\begin{align}\label{exp1}
e^{(M+dI)t} &= \sum_{k=0}^{k_0}\frac{(M+dI)^{k} t^{k}}{k!}+\sum_{k=k_0+1}^{\infty}\frac{(M+dI)^{k} t^{k}}{k!}
\end{align}
where the real part of second summation is positive. Moreover, the real part of the overall summation can be made positive by increasing the value of $t$. Thus, there exists $t_0>0$ such that $\Re (e^{(M+dI) t})>0$ for all $t\geq t_0$; meaning $(M+dI)$ is rEEP. This, in turn, implies that $M$ is rEEP.  
 
$(iii)\implies(i):$  Given that $M$ is rEEP implies that there exists a $k=k_{0}$ such that $\Re((e^M)^k)>0$ for all $k \geq k_{0}$; meaning $e^M$ is rEP. We already know that rEP is equivalent to strong complex PF property. Hence, its dominant eigenvalue, say $e^{\bar{\lambda}}$ is simple and positive. The corresponding eigenvectors $v$ and $w$ have positive real parts. Since $e^{\bar{\lambda}}$ is the maximum modulus eigenvalue, $\bar{\lambda}=\lambda_{s}$ becomes the spectral abscissa of $M$. 

For the matrix $(M+dI)$, we can always choose $d\geqslant 0$ suitably such that $(\lambda_{s}+d) > |\mu_i + d| \geqslant 0$ holds for all $\mu_i \in \operatorname{spec}(M)$. From the Jordan decomposition, it follows that the eigenspaces of $\lambda_s$ of $M$, $(\lambda_s+d)$ of $M+dI$ and $e^{\lambda_s}$ of $e^{M}$ are the same. We already know that the eigenvectors $v$ and $w$ have positive real parts, implying that $(M+dI)$ and $(M+dI)^{H}$ satisfy the strong complex PF property. 
\end{proof}
Thm. \ref{thm:stronglemma} shows that while a matrix $M$ may not by itself satisfy the strong complex PF property, by virtue of its rEEP property, it can be modified slightly to $M+dI$ to satisfy both the strong complex PF and rEP properties. A similar result holds for rENN and rEENN complex matrices that possess the complex PF property defined in $\mathcal{P}$ in Eqn.~\eqref{eq:PFmatrix}. 
\begin{theorem}\label{thm:weaklemma}
Under the hypothesis of Thm. \ref{thm:stronglemma}, the following statements are equivalent for some $d\geq 0$:
\begin{enumerate}[label=(\roman*)]
\item $ (M+dI) \text{ and } (M+dI)^{H} \in \mathcal{P}$.
\item $(M+dI)$ is rENN.
\item $M$ is rEENN. 
\end{enumerate}  
\end{theorem}
\begin{proof}
Similar to the proof of Thm.~\ref{thm:stronglemma} except that
in Eqn. \eqref{eq:Matsteadystate} we use the complex PF property in $\mathcal{P}$ instead of strong complex PF property in $\mathcal{PF}$. {\color{black}Thus, the dominant eigenvectors of $M+dI$ are real non-negative because $(M+dI)\in \mathcal{P}$; and hence, it is rENN which further implies $M$ is rEENN.} For the remaining, the strict $(>)$ inequalities in Thm. \ref{thm:stronglemma} get replaced with non-strict ones $(\geq)$. 
\end{proof}
Thms. \ref{thm:stronglemma} and \ref{thm:weaklemma} state that a complex-valued matrix $M$ is rEEP or rEENN only if the translated matrix $(M+dI)$ belongs to $\mathcal{PF}$ and $\mathcal{P}$, respectively. This key characterization is at the heart of our proofs to the results in the next section.


\section{Generalized Laplacian Matrices: Spectral and Network-Theoretic Results}\label{sec:laplacian}
Laplacian matrices are fundamental in analyzing the structure and dynamics of graphs. While the real-valued Laplacian of a non-negative digraph is diagonally dominant with real, non-negative eigenvalues, the complex-valued Laplacian (for e.g., non-Hermitian digraphs) exhibits markedly different spectral behavior, including the possibility of unstable modes (see Example~\ref{ex:unstable eigenvalues}). These differences necessitate a separate analysis, and we highlight key distinctions in graph-theoretic properties between the real and complex cases.
\begin{example}\label{ex:unstable eigenvalues}
Consider 
\begin{equation*}
L=\begin{bmatrix}
250+960\iota & 0  & -250-960\iota \\
-173-984\iota & 173+984\iota & 0\\ 
0 & -87.2-996\iota & 87.2+996\iota \\
\end{bmatrix},
\end{equation*}  
Here $\operatorname{spec}(L)=\{1107.7+1321\iota, 0, -597.5+1618.3\iota \}$ has one eigenvalue in the left-half of the complex-plane. This would not be the case if we make all the imaginary parts zero. In that case, $L$ would correspond to the Laplacian matrix of a real-valued, non-negative digraph, and hence the eigenvalues are $\operatorname{spec}(L_{real})=\{255+122\iota, 0, 255-122\iota\}$. Thus, the spectral properties of complex-valued Laplacians can differ significantly from those of their real-valued counterparts.
\qed 
\end{example}
From Sec. \ref{sec:preliminaries}, it is clear that $L\mathbb{1}_n=\mathbb{0}_n$, similar to the case of the real-valued, non-negative diagraphs. 
Hence, the right eigenvector $v$ corresponding to zero eigenvalue belongs to the $\operatorname{span}\{\mathbb{1}_n\}$, where $v=\alpha \mathbb{1}_n \text{ for all } \alpha \in \mathbb{C}$. However, both the right and left eigenvectors associated with the zero eigenvalue are of utmost importance for the characterization of complex-valued Laplacian matrices. 

Yet another property particular to complex-valued digraphs is the concept of walks. Recall that in a digraph, if $(A^k)_{ij}\ne0$, for some $k>0$, there exists a $k$-length walk between nodes $i$ and $j$. In real-valued non-negative digraphs, the reverse also holds: a $k$-length walk between $i$ and $j$ implies $(A^k)_{ij} > 0$. But it might not hold for complex-valued directed networks. To see this, let $k=2$ and note $(A^2)_{ij}=\sum_h(A)_{ih}(A)_{hj}$, where the edge-weights $A_{ih}\in \mathbb{C}$. The multiplication of complex-valued edge-weights along a directed walk may result in zero, even when a directed path exists between two nodes. Thus, the property of irreducibility in the complex-valued adjacency matrix differs from that of the real-valued non-negative adjacency matrix as discussed next.
\begin{lemma}\label{lemma:irreducible}
Consider $\mathcal{G}(A)$ with $A \in \mathbb{C}^{n\times n}$ and $n \geq 2$. Then,  
  \begin{enumerate}[label=(\roman*)]
  \item $\mathcal{G}(A)$ is strongly connected iff $A$ is irreducible
    \item If $\sum_{k=0}^{n-1} A^{k} \neq 0$, then $\mathcal{G}(A)$ is strongly connected.
    \item {\color{black} If all components of the $j^{\text {th }} \text{ column of } \sum_{k=0}^{n-1} A^{k} $ are non-zeros, then the $j^{\text {th }}$ node is globally reachable.} 
\end{enumerate}
\end{lemma}
\begin{proof}
(i) See \cite[ Thm.1.6]{vargamatrix}. 
    
(ii) We prove by contradiction. Let $\mathcal{G}(A)$ be weakly connected. So, $A$ is reducible and there exists a permutation matrix $P$ such that $PAP^{T}$ contains a zero block matrix. Thus, there does not exist a walk between some node pairs. Therefore, the summation $\sum_{k=0}^{n-1} A^{k} $ must have a zero entry, which contradicts the hypothesis that $\sum_{k=0}^{n-1} A^{k}\ne 0$.
Then, there must exist $k$ such that $(A^k)_{ij} \neq 0$ which ensures walks of some lengths between all node pairs $i$ and $j$. Thus, $\mathcal{G}(A)$ is strongly connected.

(iii) Assume that the $j^{\text {th }}$ column of $\sum_{k=0}^{n-1} A^k=I+A+A^{2}+\ldots+A^{n-1}$ is non-zero. Recall that the rows of $A$ represent outgoing edges whereas the columns represent incoming edges.    
Thus, the non-zero entries of the $j^{\text {th }}$ column of $\sum_{k=0}^{n-1}A^k$ imply existence of walks of different lengths from every other node to the $j^{\text {th }}$ node, making it globally reachable.
\end{proof}

\vspace{-0.5mm}
Lemma~\ref{lemma:irreducible} extends known results for real-valued digraphs to the complex domain, with notable distinctions. First, for a strongly connected digraph $\mathcal{G}(A)$ with complex edge-weights, the sum $\sum_{k=0}^{n-1} A^{k}$ is only guaranteed to be nonzero. Second, in the real case, statements $(ii)$ and $(iii)$ are necessary and sufficient conditions, whereas in the complex case, only the sufficiency holds. The example shown below illustrates that the reverse direction in $(iii)$ fails for complex-valued $\mathcal{G}(A)$. 
\begin{example}
Consider the complex-valued digraph $\mathcal{G}(A)$ in Fig.~\ref{fig:Weakly connected digraph with a globally reachable node} with the adjacency matrix:
\begin{figure}
    \centering
    \includegraphics[scale=0.5]{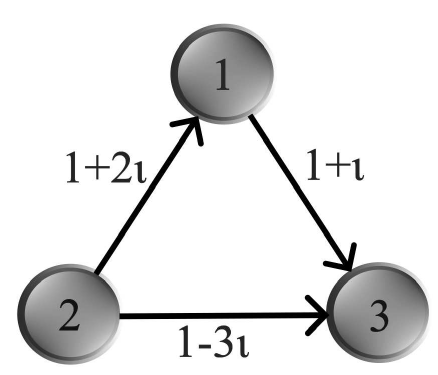}
    \caption{Weakly connected digraph with a globally reachable node.}
     \label{fig:Weakly connected digraph with a globally reachable node}
\end{figure}
\begin{equation*}
A=\begin{bmatrix}
0  & 0  & 1+\iota \\
 1+2\iota & 0& 1-3\iota \\
  0 &0 & 0\\ 
\end{bmatrix}.
\end{equation*}
The third node is globally reachable as it has a directed walk of lengths $1 \text{ and }2$ from nodes $1 \text{ and }2 \text{ to node }3$, respectively. Since $n=3$, the summation matrix is:
\begin{align*}
\sum_{k=0}^2A^k=\begin{bmatrix}
1  & 0  & 1+\iota \\
1+2\iota &1 & 0\\ 
 0 & 0 & 1\\
\end{bmatrix}. 
\end{align*}
Observe that the third column in above the sum has a zero entry. Thus, the reverse direction of statement $(iii)$ in Lemma \ref{lemma:irreducible} does not hold. \qed
\end{example} 

\vspace{-2.0mm}
\subsection{Strongly Connected Digraphs}\label{sec:strongly connected digraphs}

As discussed earlier, standard tools for analyzing consensus in real-valued non-negative digraphs do not directly extend to complex-valued strongly connected digraphs. Establishing consensus in this setting requires a broader framework, which is the focus of this section. We analyze the spectral properties of complex Laplacians $L \in \mathbb{C}^{n \times n}$ via the notion of real eventual exponential positivity (rEEP), which, as shown in Thm.~\ref{thm:stronglemma}, is tied to the \emph{real dominance} of the eigenvectors (Def.~\ref{def:real dominance}). However, not all complex matrices are rEEP or belong to the class $\mathcal{PF}$. To address this, we work with this shifted matrix
\begin{equation}\label{eq:B}
    B = dI - L, \quad \text{where } d > \max_{i=1,\dots,n} |\lambda_i(L)|. 
\end{equation}
Note that $B$ and $L$ have the same eigenspaces and the matrix $B$ plays a pivotal role in stating the conditions for marginal stability of $-L$ in our subsequent results. {\color{black}We also impose mild constraints on the phase angles of the edge weights to ensure \emph{real dominance} of the eigenvectors corresponding to the zero eigenvalue of $L$ which is discussed in the next result.} 
\begin{lemma}\label{lemma:vw}
Consider a digraph $\mathcal{G}(A)$ and let $\theta_{ij}=\angle B_{ij}$ for $1\leq i,j\leq n$. 
Let $v=\alpha \mathbb{1}_n$ (where $\alpha \in \mathbb{C}$) and $w \in \mathbb{C}^n$ be the right and left eigenvector associated with the zero eigenvalue of $L \in \mathbb{C}^{n\times n}$, respectively. Suppose $\Re(\alpha)\geq |\Im(\alpha)|$ and the phase angles \( \phi_{wi} = \angle (w)_i \) satisfy the following relation: 
    \begin{equation}\label{eq:tan}
         \left(\frac{\sum_{j=1}^{n} b_{ji} w_j \sin \left(\theta_{ji}-\phi_{wj}\right)}{\sum_{j=1}^{n} b_{ji} w_j \cos \left(\theta_{ji}-\phi_{wj}\right)}\right) \in [-1, 1] \quad \text{for all } 1 \leq i, j \leq n,
    \end{equation}
    where $(w)_i$ is the $i^{\text{th}}$ entry of $w$. Then, $v \text{ and }w$ are real dominant. 
\end{lemma}
\begin{proof}
Consider $B \text{ and }d$ as in Eqn. \eqref{eq:B}. Then, $v \text{ and }w$ are the eigenvectors corresponding to the eigenvalue $d$. If $\alpha$ is real dominant in $v=\alpha \mathbb{1}_n$, then $v$ is real dominant. Further, consider $w^H B =\lambda_1 w^H$ where $\angle\lambda_1=0^{\circ}$: 
$$
\begin{gathered}\label{eq: AppA}
\left(w^H \right)B=\lambda_1 \left(w^H \right), \\
w^H
\left[
\begin{array}{lll}
b_{11} \angle \theta_{11} & b_{12} \angle \theta_{12} & b_{13} \angle \theta_{13} \\
b_{21} \angle \theta_{21} & b_{22} \angle \theta_{22} & b_{23} \angle \theta_{23} \\
b_{31} \angle \theta_{31} & b_{32} \angle \theta_{32} & b_{33} \angle \theta_{33}
\end{array}
\right]=
\lambda_1 w^H,
\end{gathered}
$$
where $w^H=[w_1 \angle-\phi_{w1} ~ w_2 \angle -\phi_{w2} ~ w_3 \angle -\phi_{w3}]$, $b_{ij}$ and $w_i$ are the magnitudes, and $\theta_{ij}$ and $\phi_{vj}$ are the phase angles. We expand the matrix-vector product and equate the phase on both sides of the equation to get the following relation:
$$
\begin{aligned}
\operatorname{tan} \left(-\phi_{wi}\right) & =
 \left(\frac{\sum_{j=1}^{n} b_{ji} w_j \sin \left(\theta_{ji}-\phi_{wj}\right)}{\sum_{j=1}^{n} b_{ji} w_j \cos \left(\theta_{ji}-\phi_{wj}\right)}\right).
\end{aligned}
$$
Thus, if the value of $\operatorname{tan} \left(-\phi_{wi}\right)\text{ lies in }[-1,1]\text{ for every } i=1,2,...,n$ which implies \emph{real dominance} in $w$. 
\end{proof}    
{The characterization in the above lemma pertains to the weaker real dominance condition, i.e., $\Re(\alpha) \geq |\Im(\alpha)|$, which is sufficient for proving stability results in weakly connected digraphs. But to establish similar results for strongly connected digraphs, we need the \emph{strict real dominance} condition: $\Re(\alpha) > |\Im(\alpha)|$, along with the requirement that the value 
in Eqn.~\eqref{eq:tan} lies in the interval $(-1, 1)$.}
 Leveraging the \emph{strict real dominance} of eigenvectors, we finally provide the conditions for ensuring the rEEP property of Laplacian matrices. 
\begin{theorem}\label{thm:non wt.balanced}
Consider a strongly connected digraph $\mathcal{G}(A)$ and let $\theta_{ij}=\angle B_{ij}$ $($defined in Eq. \eqref{eq:B}$)$ for $1\leq i,j\leq n$. 
Let $v=\alpha \mathbb{1}_n, \alpha \in \mathbb{C} 
\text{ and } w \in \mathbb{C}^n$ be the right and left eigenvector associated with the zero eigenvalue of $L \in \mathbb{C}^{n\times n}$, respectively. Suppose that $\alpha$ and the phase angle $ \phi_{wi} = \angle (w)_i $ satisfy Lemma \ref{lemma:vw} with strict real dominance. Then, the following statements are equivalent:
  \begin{enumerate}[label=(\roman*)]
      \item $-L$ is rEEP.
      \item $-L$ is marginally stable of corank $1$.
    \end{enumerate}
\end{theorem}
The proof is in Appendix. We present a corollary that establishes the relationship between all the entries of the left eigenvector corresponding to the zero eigenvalue. This result is subsequently used to analyze weakly connected digraphs with globally reachable nodes.
\begin{corollary}\label{cor:globally reachable node}
Consider a digraph $\mathcal{G}(A)$ and  $\theta_{ij}=\angle B_{ij}$ $($defined in Eqn. \eqref{eq:B}$)$, for $1\leq i,j\leq n$. 
Let $v=\alpha \mathbb{1}_n, \alpha \in \mathbb{C} \text{ and } w \in \mathbb{C}^n$ be the right and left eigenvector associated with the zero eigenvalue of $L \in \mathbb{C}^{n\times n}$, respectively. Suppose that $\alpha$ and the phase angle $ \phi_{wi} = \angle (w)_i $ satisfy Lemma \ref{lemma:vw}. If $-L$ is rEEP and is equivalently marginally stable of corank $1$, then $\mathcal{G}(A)$ is strongly connected. 
\end{corollary}
\begin{proof}
We prove this by contradiction. Suppose that $\mathcal{G}(A)$ is weakly connected. Then, $L$ is reducible (see Lemma 1, \cite{saxena2024realeventualexponentialpositivity}). Then there exists a permutation matrix $P$ such that 
$$
\widetilde{L}=P L P^{T} = \begin{bmatrix}
L_{11} &L_{12} \\ 
0 & L_{22}
\end{bmatrix}.$$ Note that $\widetilde{L}^{l}$, for any $1\leq l\leq n$ has a zero block indicating that there does not exist a walk between certain subset of nodes.
Thus, $\exp(-Lt)\triangleq \sum_{k=0}^{\infty}\frac{(-Lt)^{k}}{k!}$ has a submatrix with all the entries as zero; this is because there do not exist walks of any length between certain node pairs. However, $-L$ is rEEP by assumption; and hence, $\operatorname{exp}(-Lt)$ has all the positive real entries. This contradicts our hypothesis that $L$ is reducible. Thus, $\mathcal{G}(A)$ is strongly connected. 
\end{proof}
\begin{example}\label{ex: Thm1_main_resultholds}
For the network shown in Fig.~\ref{fig:sc graph}, the Laplacian matrix is given in Eqn. \eqref{eq:stableL}.
    \begin{figure}
    \begin{subfigure}{0.23\textwidth}
    \centering
    \includegraphics[scale=0.479]{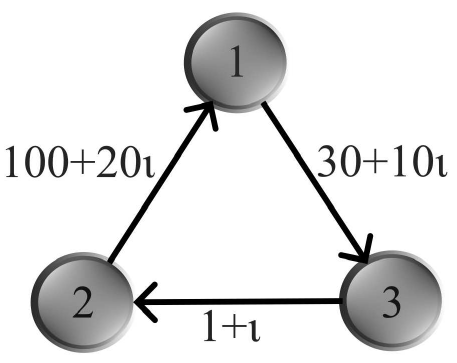}
    \caption{\small A strongly connected network with different edge weights.}
    \label{fig:sc graph}
  \end{subfigure}
    \hspace{-4mm}
  \begin{subfigure}{0.25\textwidth}
    \centering  
\includegraphics[scale=0.17]{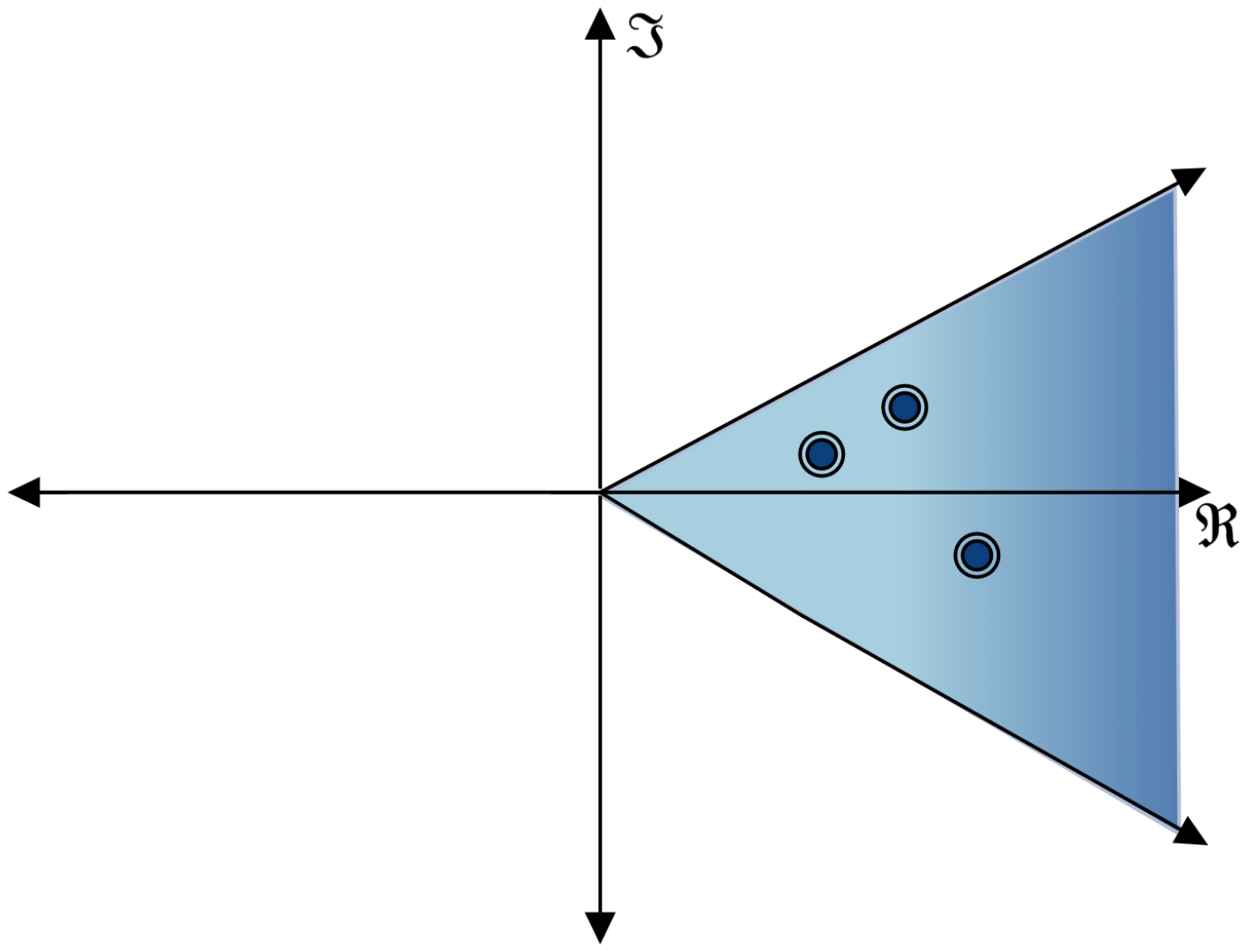}
\caption{\small Complex plane with the shaded region highlighting $\Re(a)\geq \Im(b)$ for any $z=a+b\iota$.}
\label{fig:area of real dominance}
  \end{subfigure}
  \caption{\small A strongly connected network with the right ($v$) and left ($v$) eigenvectors satisfying the assumptions in Lemma~\ref{lemma:vw}. The components of $v \text{ and }w$ are shown as dots in Fig. \ref{fig:area of real dominance}.}
  \label{fig:strongly connected digraph}
\end{figure}


\begin{equation}\label{eq:stableL}
L=\begin{bmatrix}
30+10\iota  & 0  & -30-10\iota \\
  -100-20\iota & 100+20\iota & 0\\ 
  0 & -1-\iota & 1+\iota \\
\end{bmatrix}.
\end{equation}   
Here $v=\alpha \mathbb{1}_{n}$ with $\alpha=0.577$ and $w=[0.067+0.031\iota ~  0.019+0.012\iota ~  1.646 - 0.043\iota]^{\top}$ satisfy \emph{real dominance}. Hence, for this Laplacian, the Lemma~\ref{lemma:vw} holds. Let $B=dI-L \text{ for } d =100$. Additionally, owing to \emph{real dominance} again, $B\in \mathcal{PF}$. Since the eigenvectors of $B$ and $L$ are the same, Lemma \ref{lemma:vw} holds for $v \text{ and }w$. Hence, resorting to Thm. \ref{thm:non wt.balanced}, we claim that $-L$ is rEEP, which is verified numerically by evaluating the matrix exponential at $t = 1$: 
\begin{equation*}\label{eq:reepL}
e^{-Lt}=\begin{bmatrix}
38.5+ 18\iota  & 11.2 + 7\iota  &950-25\iota\\
38.5+ 18\iota & 11.2 + 7\iota & 950-25\iota\\ 
38.5+ 18\iota & 11.2 + 7\iota & 950-25\iota \\
\end{bmatrix},
\end{equation*}
whose real parts of $e^{-Lt}$ are positive. 
\qed 
\end{example}
Ex.~\ref{ex: Thm1_main_resultholds} presents a complex-valued strongly connected digraph $\mathcal{G}(A)$ whose Laplacian satisfies the conditions presented in Lemma~\ref{lemma:vw}. As a result, $-L$ is rEEP. However, Lemma \ref{lemma:vw} does not provide a necessary condition for $-L$ to be rEEP. Ex. \ref{ex:edgemodification} shows (at least numerically) that $-L$ is rEEP even when the left eigenvector $w$ violates Eqn. \eqref{eq:tan}.
 \begin{example}\label{ex:edgemodification}
   Consider the Laplacian matrix as, 
    \begin{equation*}
     L=\begin{bmatrix}
100+14\iota  & 0  & -100-14\iota \\
  -0.55-1.92\iota & 0.55+1.92\iota & 0\\ 
  0 & -38.6-10\iota & 38.6+10\iota \\
\end{bmatrix}.   
    \end{equation*}
Here, $v=\alpha \mathbb{1}_{n}$ with $\alpha=0.577$ and $\angle(w)=[-62.6^{\circ} ~  3^{\circ} ~ -56.1^{\circ}]^{\top}$. The assumption of \emph{real dominance} is not followed in the left eigenvector corresponding to zero eigenvalue. Geometrically, imagine each element of $w$ as a vector in the complex plane. Then both the first and third components of $w$ do not lie in the region $[-45^{\circ},45^{\circ}]$  (Fig. \ref{fig:area of real dominance}). But at $t=1$ we have 
\begin{equation*}
e^{-Lt}=\setlength{\arraycolsep}{1.27pt} \begin{bmatrix}
8.8 + 17\iota  & 964 -57\iota  &  27 + 40\iota\\
8.8 + 17\iota & 964 -57\iota&27 + 40\iota\\ 
8.8 + 17\iota & 964 -57\iota & 27 + 40\iota \\
\end{bmatrix}, 
\end{equation*}
indicating that $-L$ is rEEP even if its eigenvectors are not \emph{real dominant}. \qed 
\end{example}

\vspace{-3.0mm}
\subsection{Weakly Connected Digraphs}\label{sec:weakly connected digraphs}
The previous section established equivalences between rEEP and the zero eigenvalue of $-L$ for strongly connected digraphs. We now establish analogous equivalences between rEENN and the marginal stability of $-L$ for weakly connected digraphs.
\begin{theorem} \label{thm:globally reachable}
Consider a weakly connected digraph $\mathcal{G}(A)$ and  $\theta_{ij}=\angle B_{ij}$ $($defined in Eqn. \eqref{eq:B}$)$, for $1\leq i,j\leq n$. Let $v=\alpha \mathbb{1}_n,~\alpha \in \mathbb{C} \text{ and }w \in \mathbb{C}^n$ be the right and left eigenvectors associated with the zero eigenvalue of $L \in \mathbb{C}^{n\times n}$, respectively. Suppose that $\alpha$ and the phase angle $\phi_{wi} = \angle (w)_i $ satisfy Lemma \ref{lemma:vw}.
If $L$ has corank $1$, then $\mathcal{G}(A)$ has at least a globally reachable node. Additionally, when $-L$ is marginally stable of corank $1$, then $-L$ is rEENN.  
\end{theorem}
\begin{proof}
Let $\mathcal{G}(A)$ be a complex-valued digraph with a Laplacian matrix $-L$ of corank $1$; this implies that $\mathcal{G}(A)$ has a globally reachable node. The proof is along the same lines as that of the proof of the result \cite[Thm. 7.4]{bullo2018lectures} as the latter does not depend on the values of the edge-weights of the digraph.

Next, consider the translated matrix $B=dI-L$, where $d$ is chosen according to Eqn.~\eqref{eq:B}. Since $-L$ is marginally stable of corank $1$, we know that `$d$' is the real, positive and simple dominant eigenvalue of $B$ (refer to the proof of Thm. \ref{thm:non wt.balanced}). Since the eigenvectors $v \text{ and }w $ satisfy Lemma \ref{lemma:vw}, it follows that $B$ and $B^H$ are elements of $\mathcal{P}$ (see Eqn.~\eqref{eq:PFmatrix}) and not necessarily $\mathcal{PF}$. Applying Thm. \ref{thm:weaklemma} suitably, $-L$ is rEENN.
\end{proof}
We address weakly connected digraphs without globally reachable nodes. In general, right eigenvectors corresponding to the zero eigenvalue of a $L\in \mathbb{C}$ might not lie in $\operatorname{span}\left\{\mathbb{1}_{n}\right\}$ (discussed in Thm. \ref{thm:weakly connected}). Thus, we introduce a phase angle condition on the entries of $v$ for it to exhibit \emph{real dominance}.
\begin{proposition}\label{remark:tan condition on v}
Consider a weakly connected digraph without a globally reachable node. Suppose the phase
angles $\phi_{vi} = \angle v_i$, for all $1\leq i,j\leq n$ satisfy the following relation: 
\begin{equation}\label{eq:tan condition on v}
\operatorname{tan} \left(\phi_{vi}\right)  =
 \left(\frac{\sum_{j=1}^{n} l_{ji} v_j \sin \left(\theta_{ji}+\phi_{vj}\right)}{\sum_{j=1}^{n} l_{ji} v_j \cos \left(\theta_{ji}+\phi_{vj}\right)}\right) \in [-1,1]. 
\end{equation}
Then the right eigenvector $v$ is real dominant.
\end{proposition}
\begin{proof}
The condition in Eqn. \eqref{eq:tan condition on v} has been derived by solving the eigenvector problem for a complex-valued Laplacian. Consider $Bv =\lambda_1 v$ where $\angle\lambda_1=0^{\circ}$ for an $n\times n$ matrix. The remaining proof follows on the same lines as in Lemma \ref{lemma:vw}. 
\end{proof}
\begin{theorem} \label{thm:weakly connected}
Consider a weakly connected digraph $\mathcal{G}(A)$ and let $\theta_{ij}=\angle B_{ij}$ be as in Eq.~\eqref{eq:B} for $L \in \mathbb{C}^{n\times n}$. Suppose that $\phi_{wi} = \angle (w)_i\text{ and } \phi_{vi}=\angle v_{i} $ satisfy Eqns. \eqref{eq:tan} and \eqref{eq:tan condition on v} where $\phi_{wi} \text{ and }\phi_{vi}$ are the phase angles of $i^{th}$ component of the left and right eigenvectors corresponding to the zero eigenvalue, respectively. Let $\mathcal{G}(A)$ be weakly connected with $n_{s}$ number of sinks, then $-L$ has zero as a semi-simple eigenvalue (with $n_{s}$ multiplicity). Further, a marginally stable $-L$ is marginally is rEENN. 
\end{theorem}
\begin{proof}
Since $\mathcal{G}(A)$ is weakly connected, it can have multiple sinks. Just to recall, a sink is a node with zero outgoing edges. The complex-valued Laplacian matrix may have a semi-simple zero eigenvalue depending on the number of sinks. The proof of this claim follows by adapting the real-valued counterpart given by \cite[Thm. 6.6]{bullo2018lectures}. 
Let $n_s$ be the algebraic multiplicity of the zero eigenvalue. Then at the steady state, we have
\begin{equation} \label{eq:steady state for digraphs without globally reachable node}
 \lim _{t \rightarrow \infty} e^{-Lt}= \sum_{i=1}^{n_s} v_i w_i^{H}.
\end{equation}
From Eqn.~\eqref{eq:steady state for digraphs without globally reachable node}, it suffices to show that the entries of all the eigenvectors (corresponding to zero eigenvalue) $v_i \text{ and }w_i$ have non-negative real parts and are \emph{real dominant} to establish that $-L$ is rEENN. Also, all the right eigenvectors does not lie in the $\operatorname{span}\{\mathbb{1}_n\}$.
Hence, using Prop. \ref{remark:tan condition on v}, we ensure \emph{real dominance} for the right eigenvectors. Then, along with the conditions in Eqns. \eqref{eq:tan} and \eqref{eq:tan condition on v}, the left and right eigenvectors corresponding to the semi-simple zero eigenvalue have non-negative real parts. Thus, $-L$ is rEENN from Eqn. \eqref{eq:steady state for digraphs without globally reachable node}. \end{proof}
\begin{figure}[htbp]
  \centering
  \begin{subfigure}{0.24\textwidth}
    \centering
    \includegraphics[scale=0.478]{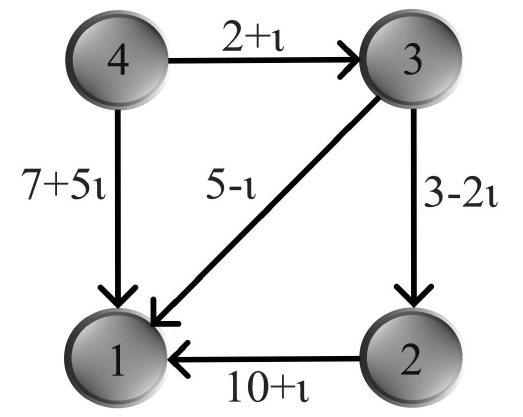}
    \caption{Weakly connected digraph with node `$1$' as a globally reachable node}
    \label{fig:digraph with a globally reachable node}
  \end{subfigure}
  \hfill
  \begin{subfigure}{0.24\textwidth}
    \centering
    \includegraphics[scale=0.478]{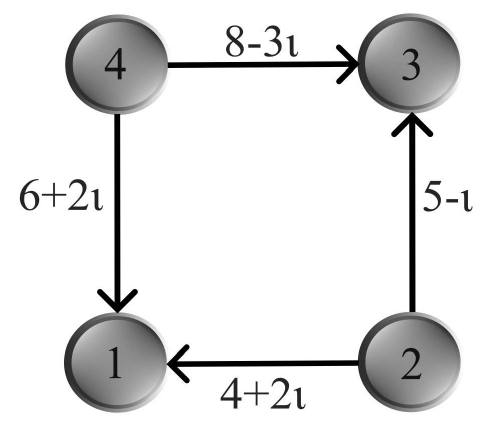}
    \caption{Weakly connected digraph with no globally reachable nodes}
    \label{fig:digraph without globally reachable node}
  \end{subfigure}
  \caption{Weakly connected digraphs}
  \label{fig:weakly connected digraphs}
\end{figure}

\begin{example}\label{ex:golbally reachable node}
Consider the weakly connected digraph as shown in Fig. \ref{fig:digraph with a globally reachable node} where node `$1$' is globally reachable. The Laplacian matrix for this digraph is 
 \begin{equation*}
 L=\begin{bmatrix}
  0& 0& 0 &0\\ 
 -10-\iota &  10+\iota &   0 &  0\\ 
 -5+\iota &  -3+2\iota  &   8-3\iota &  0\\ 
 -7-5\iota &  0 &   -2-\iota &  9+6\iota\\ 
\end{bmatrix}.    
 \end{equation*}
Here $\operatorname{spec}
(L)=\left\{ 9+6\iota, 8-3\iota,10+\iota, 0\right\}$ is of corank $1$ and the negated Laplacian is marginally stable. Thus, Thm.~\ref{thm:globally reachable} is applicable and $-L$ is rEENN. This property can be verified by examining the nonnegative real part of $e^{-L t}$ at $t=4$:
 \begin{equation*}
 e^{-L t}=\begin{bmatrix}
  1000 &  0 & 0 &0\\ 
1000 &  0 &0 &  0\\ 
1000&  0 &  0 &  0\\ 
1000 &  0 & 0 &  0\\ 
\end{bmatrix}.  
 \end{equation*}
 Moreover, since the `$1^{st }$' node is globally reachable, so we have $w=[ 2 ~ 0 ~ 0 ~ 0]^{\top}$ with its `$1^{st }$' entry real positive.  \qed 
\end{example}
\begin{example}\label{ex:weakly connected}
 Consider the weakly connected digraph in Fig.~\ref{fig:digraph without globally reachable node} without any globally reachable node. The Laplacian is:
 \begin{equation*}
 L=\begin{bmatrix}
0 & 0 & 0 & 0\\ 
-4-2\iota & 9+\iota &  -5+\iota & 0\\ 
0&  0 &  0 & 0\\ 
-6-2\iota & 0 & -8+3\iota& 14-\iota\\ 
\end{bmatrix}.    
 \end{equation*}
 Here $\operatorname{spec}(L)=\left\{ 9+\iota, 14-\iota,0, 0\right\}$ has a semi-simple zero eigenvalue. Hence, Thm.~\ref{thm:weakly connected} is applicable and $-L$ is rEENN. This property can verified by examining the real part of $e^{-L t}$ at $t=1$, which has nonnegative entries: 
 \begin{equation*}
 e^{-L t}=\begin{bmatrix}
  1000 &  0 & 0 &0\\ 
 463.4+170.8\iota &  0.1-0.1\iota &   536.6-170.7\iota &  0\\ 
 0 &  0 &   1000 &  0\\ 
 416.2+172.6\iota &  0 &   583.8-172.6\iota &  0\\ 
\end{bmatrix}.    
 \end{equation*}
\end{example} \qed

\section{Consensus in complex-valued Laplacian Flows}\label{sec:consensus}
With the network-theoretic results on weakly and strongly connected digraphs established in the previous section, we proceed to the main result of this study, focusing on consensus in complex-valued digraphs. We write the Jordan matrix $J=[0] \oplus J_\text{reduced}$, where $J_\text{reduced}$ comprises non-zero eigenvalues. Note that $\oplus$ is the direct sum of matrices, and in our case $J$ is merely a $2\times 2$ block matrix with right-lower block given by $J_\text{reduced}$.
\begin{theorem}  
 
 \label{lemma:consensus} Consider the flow system in Eqn.~\eqref{eq:Lflows} with $\bar{L}=L$. Let $L \in \mathbb{C}^{n \times n}$ has corank $1$ with right and left eigenvectors $v$ and $w$ satisfying Lemma \ref{lemma:vw}. Consensus is guaranteed in the digraph $\mathcal{G}(A)$ iff the $J_{reduced}$ has all the eigenvalues in the ORHP of the complex-plane. 
\end{theorem}
\begin{proof}
From Thm.~\ref{thm:globally reachable}, $L \in \mathbb{C}^{n \times n}$ has corank $1$ implies that the digraph $\mathcal{G}(A)$ has at least one globally reachable node. 
From Thm.~\ref{thm:non wt.balanced}, observe that $-L$ is marginally stable of corank $1$ for strongly connected digraphs, provided that $-L$ is rEEP. Similarly, from Thm.~\ref{thm:globally reachable}, $-L$ is rEENN.

Consider the solution of the flow system $\dot{x} =-Lx$, given by $x(t) = e^{-Lt} x_0$, where $x_0 \in \mathbb{C}^n$ is the vector of initial conditions. Then, by expanding $L=V([0] \oplus J_\text{reduced})W^H$, we compute the matrix exponential $e^{-Lt}$ as
 \begin{equation*}
 e^{-L t} = \left[{v} \mid V_{n, n-1}\right]\left[\begin{array}{c|c}
1 & 0 \\
\hline 0 & e^{-J_\text{reduced} t}
\end{array}\right]\left[\begin{array}{c}
{w}{^{H}} \\
\hline W_{n-1, n}^H
\end{array}\right], 
\end{equation*}
where $v$ and $w$ are the right and left eigenvectors corresponding to $\lambda_1(-L)=0$ (or $e^{\lambda_1(-L)}=1$). 

\textbf{(if)} Let $\text{spec}(J_\text{reduced})$ lie in the ORHP of the complex plane and $e^{-J_{reduced}t} \to 0$ as $t\to \infty$. Thus 
 $\lim_{t \to \infty} e^{-L t}= vw=\alpha\mathbb{1}_{n} {w}^{H}$, where $\alpha $ follows Lemma \ref{lemma:vw}. 
On the other hand, we already know that the negated Laplacian matrix is rEEP and rEENN in strongly connected and weakly connected digraphs, respectively (follows from Thms. \ref{thm:non wt.balanced} and \ref{thm:globally reachable}).
Thus,
\begin{equation} \label{Eqn:steady state of L flow}
\lim_{t \to \infty}{x}(t)= \lim_{t \to \infty} e^{-L t}x_0=\mathbb{1}_{n} (\alpha{w}^{H}{x_{0})}.
\end{equation}
The value of $\alpha {w}^{H}{x_{0}}$ is constant, and hence, $x(t) \in \operatorname{span}\left\{\mathbb{1}_{n}\right\}$ when $t \to \infty$. Hence, the consensus is achieved. 
 
(\textbf{only if}) Suppose the Laplacian flow given by Eqn.~\eqref{eq:Lflows} converges to consensus. This convergence to a non-zero value occurs only if the corresponding negated Laplacian matrix is marginally stable and of corank $1$ (see \cite[Thm. 7.1]{bullo2018lectures}). Consequently, $J_{\text{reduced}}$ must have its eigenvalues in the ORHP.
\end{proof}
\begin{remark}\label{remark:no consensus in weakly connected digraphs}
 In contrast to the strongly connected and weakly connected digraphs with at least a globally reachable node, the complex-valued Laplacian flows corresponding to weakly connected digraphs without a globally reachable node does not necessarily achieve consensus.
This is because of the multiple number of sinks in the digraph, which is equal to the number of zero eigenvalues of the Laplacian matrix. \qed 
\end{remark}
\subsection{Modified Flows}\label{sec: algoirthm}
Thm.~\ref{lemma:consensus} provides us a necessary and sufficient condition for the flow system in Eqn.~\ref{eq:Lflows} to attain consensus when $\bar{L} = L$. Specifically, the condition requires the spectrum of the reduced Jordan block $J_{\text{reduced}}$ to lie in the ORHP of the complex plane. We now consider the case where the spectrum (or a subset of it) lies in the OLHP of s-plane. 

The state trajectories in Eqn.~\ref{eq:Lflows} diverge when some or all eigenvalues lie in the OLHP, making convergence impossible. To address this, we propose modifying the network structure so that the spectrum of the new $J_{\text{reduced}}$ (corresponding to the modified matrix) lies in the ORHP, thereby enabling the application of Thm.~\ref{lemma:consensus}. 

\begin{theorem}\label{thm: eigenspectrum_assignment}{Let the digraph $\mathcal{G}(A)$ has at least one globally reachable node. Let $L=V([0] \oplus J_\text{reduced})W^H$ with the non-zero eigenvalues in both OLHP and ORHP of the complex plane. Then, the modified matrix $L_m=V([0] \oplus J'_\text{reduced})W^H\in \mathbb{C}^{n\times n}$ associated with a digraph $\mathcal{G}(A')$ has all the non-zero eigenvalues in the ORHP of the complex plane if there exists an $S \in \mathbb{C}^{n \times n} \text{ such that }J'_\text{reduced}=SJ_{reduced}$.} 
\end{theorem}
\begin{proof}
    Since the Jordan block $J_\text{reduced}$ satisfies the hypothesis of Lemma~\ref{lemma:existence of D}, there exists a complex-valued diagonal matrix ${S} \in \mathbb{C}^{n \times n }$ such that the spectrum of $J'_\text{reduced}={S}J_\text{reduced}$ lies in the ORHP. Second, the normalized right and left eigenvectors of $L$ and $L_m$ coincide, as does the simple zero eigenvalue. This implies that the modified digraph also contains a globally reachable node (by an argument similar to that in Thm.~\ref{thm:globally reachable}). Combining these observations, the conclusion of the theorem follows.
\end{proof}

\vspace{-3.0mm}

\begin{lemma}({\cite{ballantine1970stabilization}})\label{lemma:existence of D}
    Consider a matrix $M \in \mathbb{C}^{q \times q}$ whose leading principal minors are non-zero. Then, there exists a diagonal matrix ${S} \in \mathbb{C}^{q \times q}$ such that all the eigenvalues of ${S}M$ are positive and simple. 
\end{lemma}

Thm.~\ref{thm: eigenspectrum_assignment} is reminiscent of a pole placement algorithm using a prescribed set of eigenvectors—in our case, the eigenvectors of the original (unmodified) Laplacian $L$. This implies that the modified matrix (denoted as $L_m$) always has a simple zero eigenvalue with the right eigenvector $v=\alpha \mathbb{1}_n$ by virtue of Lemma \ref{lemma:existence of D}. {However, it may or may not be a Laplacian matrix, as a Laplacian matrix must have the diagonal entries being equal to the sum of off-diagonal entries. If it is a Laplacian matrix, then the edge connections in the modified digraph $\mathcal{G}(A')$ (where $A'$ is the modified adjacency matrix) may differ from those in the original digraph (see Ex.~\ref{ex:edgeaddition}). However, by construction, the modified digraph inherits the globally reachable properties of the nodes in the original digraph.} Designing an eigenspectrum assignment which modifies only specific edges remains an interesting research direction, which we leave for future work.

We present Alg.~\ref{alg:buildtree} based on Thms.~\ref{lemma:consensus} and ~\ref{thm: eigenspectrum_assignment} to systematically check all the necessary and sufficient conditions in order to guarantee consensus in complex-valued digraphs. 


\begin{algorithm}[htbp]
\caption{ Consensus in Complex-valued Laplacian Flows}
\label{alg:buildtree}
\begin{algorithmic}[1]
    \STATE \textbf{Input} $\mathcal{G}(A)$ with at least one globally reachable node, $v=\alpha \mathbb{1}_n\text{ and }w\in \mathbb{C}^n$ are the right and left eigenvectors associated with the zero eigenvalue of $\bar{L}=L \in \mathbb{C}_{n\times n}$.  
    \STATE \textbf{Output} Agents converge to consensus in the original or modified digraph.
    \WHILE{number of nodes $n \geq 2$}
            \IF{{$v \text{ and } w$ are \emph{real dominant} and }
            $J_{\text{reduced}}$ (of $-L$) has eigenvalues in ORHP;}
                \STATE States in Eqn.~\eqref{eq:Lflows} attain consensus.
            \ELSE
                \STATE Find a diagonal matrix ${S}$ with $\operatorname{spec}({S} J_{\text{reduced}})$ lie in the ORHP. Define the modified matrix: $L_m=V[[0] \oplus [{S} J_{\text{reduced}}]]W^H$ and set $\bar{L}=L_m$. Flow system in Eqn.~\eqref{eq:Lflows} attain consensus.
            \ENDIF
    \ENDWHILE
\end{algorithmic}
\end{algorithm}
\begin{example}\label{ex:edgeaddition}
Consider the Laplacian matrix: 
\begin{equation*}
L=\begin{bmatrix}
250+960\iota & 0  & -250-960\iota \\
-173-984\iota & 173+984\iota & 0\\ 
0 & -87.2-996\iota & 87.2+996\iota \\
\end{bmatrix}.
\end{equation*}
associated with the digraph shown in Fig. \ref{fig:unstable nw}. Here, $\operatorname{spec}(L)=\{1107.7+132\iota, 0, -597.5+1618.3\iota \}$ which has an eigenvalue in the OLHP of the complex plane. Thus, the state trajectories of the flow system $\dot{x}(t)=-Lx(t)$ diverge. 

The Jordan matrix of $L$ is $J=[0] \oplus J_{\text{reduced}} \text{ where } J_{\text{reduced}}=\operatorname{diag}(1107.7+132\iota, -597.5+1618.3\iota)$. Define the diagonal matrix ${S}\!=\!\operatorname{diag}( 0.524-0.625\iota ,-0.245-0.663\iota)$. Substitute ${S}J_{\text{reduced}}=\operatorname{diag}(1408,1219)$ 
in the modified matrix $L_{m}=V([0] \oplus {S}J_{reduced})W^H$ to get
\begin{equation*}
L_{m}=\begin{bmatrix}
878.7-36.8\iota & -438.2+53.31\iota & -440.5-16.5\iota \\
-444.4-90.7\iota & 876+0.28\iota & -431+90\iota\\ 
-444.2+18\iota & -428-55\iota & 873+37\iota \\
\end{bmatrix}.
\end{equation*} 
\begin{figure}[htbp]
    \begin{subfigure}{0.23\textwidth}
    \centering
    \includegraphics[scale=0.478]{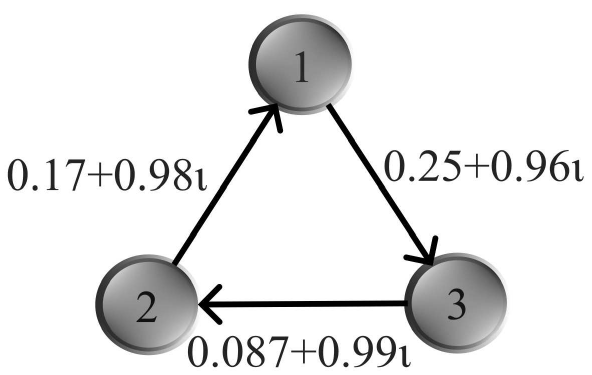}
    \caption{\small original network with a divergent Laplacian flow as it violates Thm. \ref{lemma:consensus}}
    \label{fig:unstable nw}
  \end{subfigure}
    \hfill
  \begin{subfigure}{0.23\textwidth}
    \centering
    \includegraphics[scale=0.453]{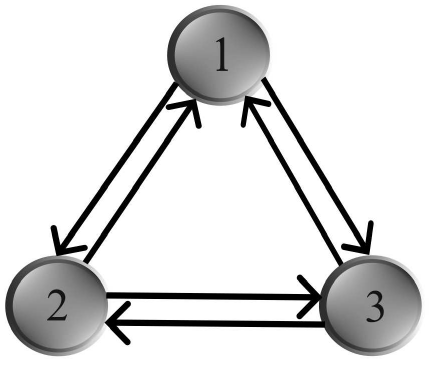}
    \caption{\small network with modified edge-weights yields convergent Laplacian flow (see Thm. \ref{thm: eigenspectrum_assignment}).}
    \label{fig:modified graph}
  \end{subfigure}
  \caption{Strongly connected digraphs with divergent and convergent Laplacian flows shown in Figs. \ref{fig:unstablenw1sim.} and \ref{fig:stablenwsim}, respectively.}
  \label{fig:edge modification}
\end{figure}
Here $\operatorname{spec}(L_{m})=\left\{ 0,1219,1408  \right\}$, which implies that $-L_{m}$ is marginally stable of corank $1$. Also, it can be checked that $w=[0.58 ~0.57+0.046\iota ~0.56+0.096\iota]^{\top}$ satisfies \emph{real dominance}. Consequently,  
the consensus in modified matrix flow system is guaranteed by Thm. \ref{lemma:consensus} shown in Fig. \ref{fig:stablenwsim}.   \qed
\end{example}

\vspace{-5.0mm}
\subsection{Diffusion Dynamics}\label{sec:Applications}
 {\color{black}Diffusion in social networks enables to us understand how information spreads and to find the most influential agent in the multi-agent system (see \cite{shrinate2024centrality,veerman2019diffusion} for details).
In real-valued networks, \cite{veerman2019diffusion} showed that diffusion and consensus are dual dynamical processes. Indeed, diffusion is merely the consensus for the adjoint system $\dot{p}^\top=-p^\top \mathcal{L}$, where $\mathcal{L}\triangleq I-D^{-1}A \in \mathbb{R}^{n\times n}$ is the symmetric random walk Laplacian matrix, $D$ is the degree matrix, and $p$ is the probability vector with strictly positive entries that sum to unity \cite{masuda2017random}. 

Extending diffusion dynamics to complex-valued networks is straightforward by defining $\mathcal{L}=I-D^{-1}A$, where $D$ and $A$ are complex. For weight-balanced and Hermitian digraphs, one can show that \( \mathcal{L} \) is real-valued, ensuring the existence of a probability vector as discussed above; see Ex.~\ref{ex:diffusion-example} for an illustration. However, for general complex-valued networks, \( \mathcal{L} \) may be complex, and the probabilistic interpretation of the diffusion dynamics no longer holds in the conventional sense. A possible alternative interpretation is proposed in~\cite{tian2024structuralbalance}.

Diffusion dynamics are used to identify the most influential agent. To this aim, we introduce the notion of the \emph{influence vector}} for weight-balanced and Hermitian digraphs \cite{masuda2010dynamics}: 
\begin{equation}\label{eq:influence}
\mathcal{I}=\frac{\mathbb{1}_{n}^{T}}{n} v w^{H}\in \mathbb{R}^n, 
\end{equation}
{\color{black} where $v$ and $w$ are the normalized right and left eigenvectors corresponding to zero eigenvalue of the Laplacian matrix.}. There is a direct relationship between influence and steady-state diffusion probabilities: nodes with greater influence in $\mathcal{I}$ achieve higher steady-state probability values. 
\begin{example}\label{ex:diffusion-example}
Consider a weight-balanced digraph with
 \begin{equation*}   
L=\begin{bmatrix}
1+0.5\iota  & 0 &  -1-0.5\iota \\
-1-0.5\iota & 1+0.5\iota & 0\\ 
 0& -1-0.5\iota & 
1+0.5\iota\\
\end{bmatrix}.
\end{equation*}  
The random walk Laplacian matrix $\mathcal{L}$ is 
\begin{equation*}   
\mathcal{L}=\begin{bmatrix}
1 & 0 &  -1 \\
-1 & 1 & 0\\ 
 0& -1 & 1\\
\end{bmatrix}.
\end{equation*}  
whose eigenvectors corresponding to the zero eigenvalue are 
$v=[0.5773~ 0.5773~ 0.5773]^{\top}, w=[0.5773~ 0.5773~ 0.5773]^{\top}$. The influence vector given in Eqn. \eqref{eq:influence} evaluates to 
\begin{equation*}   
\mathcal{I}=\frac{\mathbb{1}_{n}^{T}}{n} vw^{H}=\begin{bmatrix}
0.333 & 0.333 &  0.333
\end{bmatrix}.
\end{equation*}  
Since the digraph is weight-balanced, it is not at all surprising that all the components of $\mathcal{I}$ are identical, resulting in equally influential nodes. \qed 
\end{example}



\section{Simulation Results}\label{sec:results}
We verify our mathematical findings by simulating Laplacian flows for a few representative examples discussed earlier to fairly critique our goal of achieving consensus in Thm. \ref{lemma:consensus}. To ensure that states of the Laplacian flow system (in Eqn. \eqref{eq:Lflows}) do not converge to zero, we choose initial conditions that are not orthogonal to the left eigenvector of the Laplacian matrix corresponding to the zero eigenvalue (in Eqn. \eqref{Eqn:steady state of L flow}). Since the state is complex-valued in our framework, we plot the real and imaginary state trajectories separately in all the plots.

\vspace{-3.0mm}

\subsection{Synthetic Results} 
{We present simulations that illustrate our results as well as the cases which show that the stated conditions are not necessary to achieve consensus in a complex-weighted network. We first validate our theory on a strongly connected digraph whose Laplacian has a left eigenvector satisfying the \emph{ real dominance} condition in Eqn.~\eqref{eq:tan}. Additionally, the Laplacian matrix also follows the conditions stated in Thm. \ref{lemma:consensus} which is discussed in Ex. \ref{ex: Thm1_main_resultholds}. Hence, the corresponding Laplacian flows converge to consensus\footnote{In the interest of brevity, we omit the corresponding simulation.}. 
{ Now, we consider Ex. \ref{ex:edgemodification} where the left eigenvector of Laplacian matrix does not satisfy \emph{real dominance}. The corresponding simulation shown in Fig. \ref{fig:sufficientcondnsim} reveals that the Laplacian flows still reach consensus even though the left eigenvector violates the \emph{real dominance} condition (i.e., it does not satisfy Lemma \ref{lemma:vw}). 
Further, we provide the necessary conditions for achieving consensus in Thm. \ref{lemma:consensus}. Ex. \ref{ex:edgeaddition} demonstrates a network that does not meet these necessary and sufficient conditions.} As expected, Fig. \ref{fig:unstablenw1sim.} shows that the resulting Laplacian flow is unstable. However, the modified flow system achieves consensus, as shown in Fig. \ref{fig:stablenwsim}. This behavior is predicted in Thm. \ref{thm: eigenspectrum_assignment} due to the proposed modified flows.}

We turn our attention to weakly connected digraphs. Fig. \ref{fig:grgraphsim} shows that the digraph with a globally reachable node in Ex. \ref{ex:golbally reachable node} attains consensus, which confirms the theoretical prediction in Thm. \ref{lemma:consensus}. Instead, states associated with the weakly connected digraph with two sinks in Ex. \ref{ex:weakly connected} do not converge, which is shown in Fig. \ref{fig:wcgraphsim} (see our discussion in Rmk. \ref{remark:no consensus in weakly connected digraphs}).
\begin{figure*}[htbp]
  \centering
  \begin{subfigure}{0.31\textwidth}
    \centering
\includegraphics[scale=0.41]{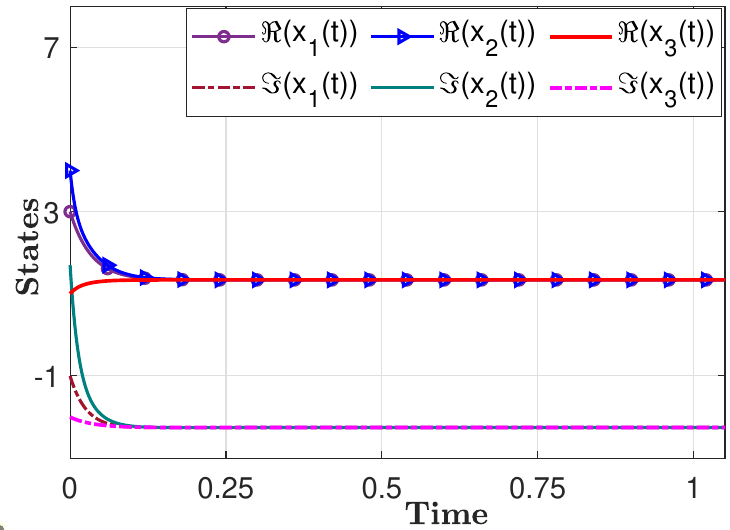}
\caption{\small States reach consensus in Ex.~\ref{ex:edgemodification}, despite the sufficient condition in Eq.~\eqref{eq:tan} being violated.}
    \label{fig:sufficientcondnsim}
  \end{subfigure}
  \begin{subfigure}{0.31\textwidth}
    \centering
    \includegraphics[scale=0.41]{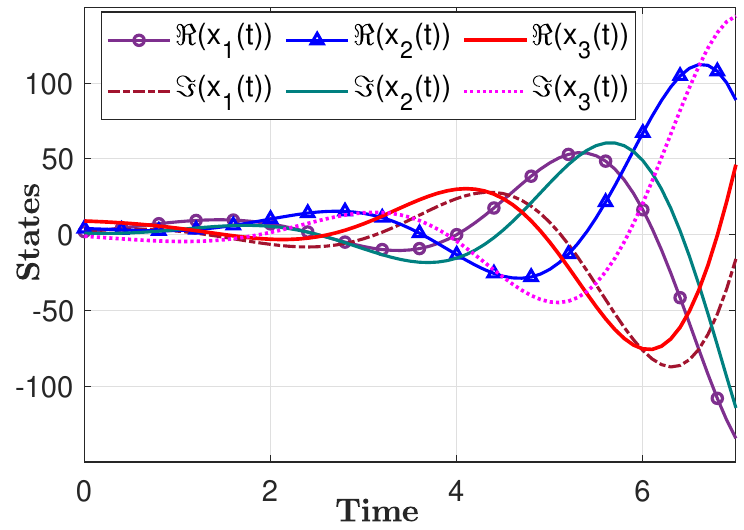}
    \caption{\small States diverge for the strongly connected network in Ex.~\ref{ex:edgeaddition} due to the violation of Thm. \ref{lemma:consensus}.}
    \label{fig:unstablenw1sim.}
  \end{subfigure}
  \begin{subfigure}{0.31\textwidth}
    \centering
    \includegraphics[scale=0.41]{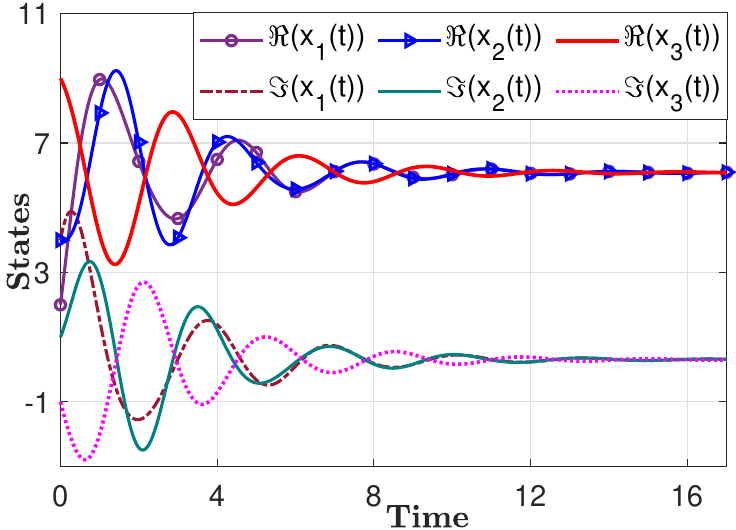}
    \caption{\small States reach consensus in the modified flow (unstable) system in Ex.~\ref{ex:edgeaddition} (cf.~Thm. 
    \ref{thm: eigenspectrum_assignment}.)}
    \label{fig:stablenwsim}
  \end{subfigure}
  \centering
  \begin{subfigure}{0.31\textwidth}
    \centering
    \includegraphics[scale=0.41]{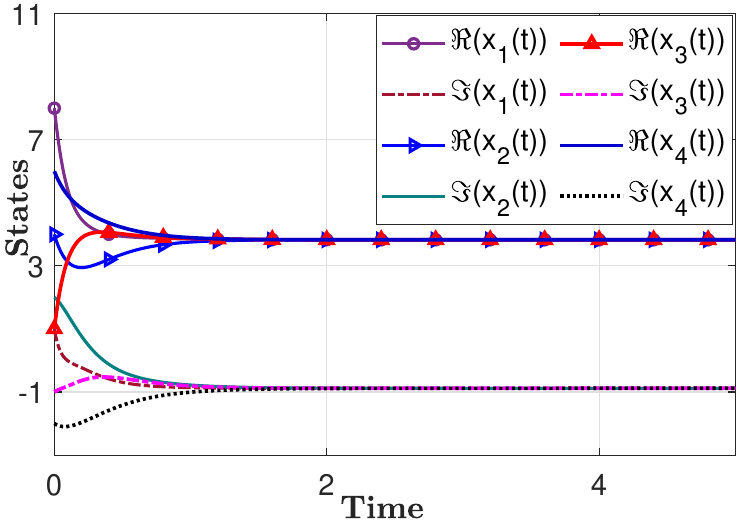}
    \caption{\small States reach consensus in a weakly connected network with a globally reachable node in Ex.~\ref{ex:golbally reachable node}.}
    \label{fig:grgraphsim}
  \end{subfigure}
  \begin{subfigure}{0.31\textwidth}
    \centering
   \includegraphics[scale=0.41]{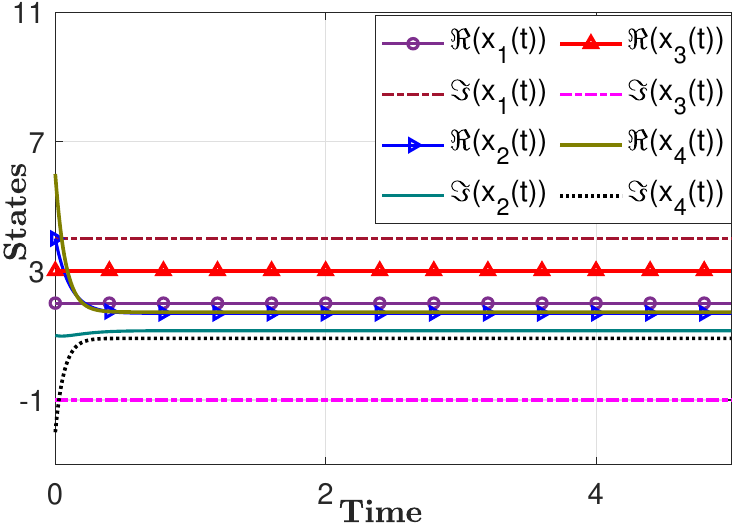}
    \caption{\small States fail to achieve consensus in a weakly connected network in Ex.~\ref{ex:weakly connected} as stated in Rmk. \ref{remark:no consensus in weakly connected digraphs}.}
    \label{fig:wcgraphsim}
  \end{subfigure}
  \begin{subfigure}{0.31\textwidth}
    \centering
    \includegraphics[scale=0.41]{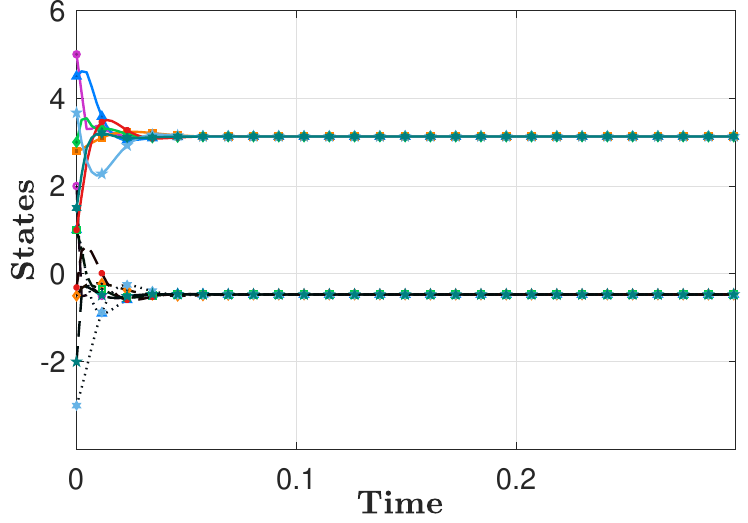}
      \caption{\small Laplacian flows corresponding to a communication network of EIES Dataset with $7$ agents. }
    \label{fig:fixednwsim}
  \end{subfigure}
  \caption{\small Steady-state behavior of complex-valued Laplacian flows in strongly and weakly connected digraphs. The solid and dotted lines represent the real and imaginary part of the states, respectively. }
  \label{fig:examplessim2}
\end{figure*} 
\subsection{Electronic Information Exchange System (EIES) }
We consider a communication network which is represented as a complex-valued digraph in \cite{hoser2005eigenspectral}. In this network, nodes represent network members, and links are formed whenever authors are linked by a message. Specifically, the real part of complex-valued edge weight signifies the number of messages from nodes $i$ to $j$, whereas the imaginary part represents the number of messages from nodes $j$ to $i$. The values of real and imaginary part of the edge weight are the number of messages between two members. 
We constructed the Laplacian matrix of the communication network. This Laplacian has a simple zero eigenvalue with remaining eigenvalues in the ORHP of the complex plane. Additionally, the eigenvectors corresponding to zero eigenvalue satisfy \emph{real dominance}. Consequently using Thm. \ref{lemma:consensus}, the real and imaginary part of all the states converge to consensus as shown in Fig. \ref{fig:fixednwsim}.
\section{Conclusion}\label{sec:conclusions}
This paper develops network and spectral theoretic methods for analyzing flow systems in complex-weighted digraphs, with the aim of developing sufficient and necessary conditions for achieving consensus. Towards this goal, we introduce the concepts of rEEP and rEENN for complex matrices and show that the negated complex-valued Laplacian is rEEP for strongly connected digraphs and rEENN for weakly connected digraphs.

In contrast to real-valued non-negative digraphs, where connectivity guarantees consensus, our analysis reveals that complex-weighted networks require additional algebraic constraints on the phase angles of the edge weights. Further, we propose an algorithm to achieve consensus by modifying the Laplacian matrix when the original network does not meet our stated conditions. Finally, we also explore a dual aspect of consensus, namely, diffusion in networks. We present illustrative examples throughout the paper, along with simulation studies on synthetic systems and a benchmark electronic information exchange system, to validate our theoretical results.

\section{Appendix}\label{sec:appendix}
\textit{Proof of Thm.~\ref{thm:non wt.balanced}:} $i) \implies ii)$ {We recall that $-L$ is said to be rEEP if and only if $\Re(\exp (-Lt))>0$ for all $t \geq t_{0}$. Using Thm. \ref{thm:stronglemma}, there exists a positive number $d$ such that the translated matrix $B$, defined in Eqn. \eqref{eq:B}, belongs to the set $\mathcal{PF}$ with the dominant eigenvector $\Re{(v)}>0$. From Lemma \ref{lemma:PF and cone}, the matrix $B$ is a cone contraction. So the dominant eigenvector $v$ must lie in the positive orthant (all positive entries) and is also unique (follows from \cite[Thm.~3.6]{rugh2010cones}).  

Since $B\in \mathcal{PF}$, the spectral abscissa $\lambda_s$ of $-L$ is simple as well which follows from the part (iii) implies (ii) of the proof of Thm.\ref{thm:stronglemma}. Further, $v$ is the eigenvector associated with $\lambda_s$ of $-L$. Additionally, $L\mathbb{1}_n=0$ which implies that $\mathbb{1}_n$ is the eigenvector associated with the eigenvalue zero. But $v$ must lie uniquely in the positive orthant implying that $v= \alpha \mathbb{1}_n$; equivalently, the spectral abscissa $\lambda_s$ must be equal to zero. Hence, $-L$ is marginally stable of corank $1$.}

$ii) \implies i)$
Let $-L$ be marginally stable of corank $1$ and construct translated matrix $B$ as given by Eqn. \eqref{eq:B}. Since zero is the spectral abscissa of $-L$, the eigenvalue $d$ is dominant. The right and left eigenvectors, $v \text{ and }w$, of $B$ are same as the eigenvectors corresponding to the zero eigenvalue of $L$. Note that both the eigenvectors already satisfy \emph{strict real dominance}. Hence, using Thm. \ref{thm:stronglemma} and Lemma \ref{lemma:vw}, we conclude that $B \in \mathcal{PF}$. Thus, $-L$ is rEEP follows from Thm. \ref{thm:stronglemma}. 
\qed
\bibliography{main}

\begin{thebibliography}{10}

\bibitem{fontan2022multiagent}
Angela Fontan, Lingfei Wang, Yiguang Hong, Guodong Shi, and Claudio Altafini.
\newblock Multiagent consensus over time-invariant and time-varying signed digraphs via eventual positivity.
\newblock {\em IEEE Transactions on Automatic Control}, 68(9):5429--5444, 2022.

\bibitem{fontan2021properties}
Angela Fontan and Claudio Altafini.
\newblock On the properties of {L}aplacian pseudoinverses.
\newblock In {\em 2021 60th IEEE Conference on Decision and Control (CDC)}, pages 5538--5543. IEEE, 2021.

\bibitem{olfati2007consensus}
Reza Olfati-Saber, J~Alex Fax, and Richard~M Murray.
\newblock Consensus and cooperation in networked multi-agent systems.
\newblock {\em Proceedings of the IEEE}, 95(1):215--233, 2007.

\bibitem{Stevenlow}
Steven~H. Low.
\newblock Power system analysis - analytical tools and structural properties, 2025.
\newblock Rough draft.

\bibitem{anderson2013network}
Brian~DO Anderson and Sumeth Vongpanitlerd.
\newblock {\em Network analysis and synthesis: a modern systems theory approach}.
\newblock Courier Corporation, 2013.

\bibitem{seshu1961reed}
Sundaram Seshu and Myril~B Reed.
\newblock {\em Linear Graphs and Electrical Networks}.
\newblock Addison Wesley, 1961.

\bibitem{wang2024first}
Dan Wang, Wei Chen, and Li~Qiu.
\newblock The first five years of a phase theory for complex systems and networks.
\newblock {\em IEEE/CAA Journal of Automatica Sinica}, 11(8):1728--1743, 2024.

\bibitem{tian2024structuralbalance}
Yu~Tian and Renaud Lambiotte.
\newblock Structural balance and random walks on complex networks with complex weights.
\newblock {\em SIAM Journal on Mathematics of Data Science}, 6(2):372--399, 2024.

\bibitem{gong2021directed}
Xue Gong, Desmond~J Higham, and Konstantinos Zygalakis.
\newblock Directed network {L}aplacians and random graph models.
\newblock {\em Royal Society open science}, 8(10):211144, 2021.

\bibitem{kubota2021quantum}
Sho Kubota, Etsuo Segawa, and Tetsuji Taniguchi.
\newblock Quantum walks defined by digraphs and generalized hermitian adjacency matrices.
\newblock {\em Quantum Information Processing}, 20(3):95, 2021.

\bibitem{mukai2020discrete}
Kanae Mukai and Naomichi Hatano.
\newblock Discrete-time quantum walk on complex networks for community detection.
\newblock {\em Physical Review Research}, 2(2):023378, 2020.

\bibitem{mn2015continuous}
Dheeraj MN and Todd~A Brun.
\newblock Continuous limit of discrete quantum walks.
\newblock {\em Physical Review A}, 91(6):062304, 2015.

\bibitem{godsil2023discrete}
Chris Godsil and Hanmeng Zhan.
\newblock {\em Discrete quantum walks on graphs and digraphs}, volume 484.
\newblock Cambridge University Press, 2023.

\bibitem{muranova2020electrical}
Anna Muranova and Robert Schippa.
\newblock Eigenvalues of the normalized complex {L}aplacian on finite electrical networks.
\newblock {\em arXiv preprint arXiv:2012.12759}, 2020.

\bibitem{kobayashi2010nnl}
Masaki Kobayashi.
\newblock Exceptional reducibility of complex-valued neural networks.
\newblock {\em IEEE Transactions on Neural Networks}, 21(7):1060--1072, 2010.

\bibitem{tong2019mimo}
Di~Tong, Yuehua Ding, Yonggui Liu, and Yide Wang.
\newblock A {MIMO}-{NOMA} framework with complex-valued power coefficients.
\newblock {\em IEEE Transactions on Vehicular Technology}, 68(3):2244--2259, 2019.

\bibitem{koteswar2023fccns}
Saurabh Yadav and Koteswar~Rao Jerripothula.
\newblock {FCCNS}: Fully complex-valued convolutional networks using complex-valued color model and loss function.
\newblock In {\em Proceedings of the IEEE/CVF International Conference on Computer Vision}, pages 10689--10698, 2023.

\bibitem{krawciw2024small}
Benjamin Krawciw, Lincoln~D Carr, and Cecilia~Diniz Behn.
\newblock The small-world effect for interferometer networks.
\newblock {\em Journal of Physics: Complexity}, 5(2):025016, 2024.

\bibitem{bottcher2024complex}
Lucas B{\"o}ttcher and Mason~A Porter.
\newblock Complex networks with complex weights.
\newblock {\em Physical Review E}, 109(2):024314, 2024.

\bibitem{li2017consensus}
Qiao Li, David~Wenzhong Gao, Huaguang Zhang, Ziping Wu, and Fei-yue Wang.
\newblock Consensus-based distributed economic dispatch control method in power systems.
\newblock {\em IEEE {T}ransactions on {S}mart {G}rid}, 10(1):941--954, 2017.

\bibitem{zhang2020consensus}
Bowen Zhang, Yucheng Dong, Hengjie Zhang, and Witold Pedrycz.
\newblock Consensus mechanism with maximum-return modifications and minimum-cost feedback: A perspective of game theory.
\newblock {\em European Journal of Operational Research}, 287(2):546--559, 2020.

\bibitem{pham2019consensus}
Thiem~V Pham, Nadhir Messai, and Noureddine Manamanni.
\newblock Consensus of multi-agent systems in clustered networks.
\newblock In {\em 2019 18th European control conference (ECC)}, pages 1085--1090. IEEE, 2019.

\bibitem{garindistributed}
F~Garin and L~Schenato.
\newblock Distributed estimation and control applications using linear consensus algorithms, {V}olume {N}etworked {C}ontrol {S}ystems.
\newblock {\em of Lecture Notes in Control and Information Sciences}.

\bibitem{dorfler2012synchronization}
Florian Dorfler and Francesco Bullo.
\newblock Synchronization and transient stability in power networks and nonuniform {K}uramoto oscillators.
\newblock {\em SIAM Journal on Control and Optimization}, 50(3):1616--1642, 2012.

\bibitem{reff2012spectral}
Nathan Reff.
\newblock Spectral properties of complex unit gain graphs.
\newblock {\em Linear algebra and its applications}, 436(9):3165--3176, 2012.

\bibitem{lin2013leader}
Zhiyun Lin, Wei Ding, Gangfeng Yan, Changbin Yu, and Alessandro Giua.
\newblock Leader--follower formation via complex {L}aplacian.
\newblock {\em Automatica}, 49(6):1900--1906, 2013.

\bibitem{dong2014complex}
Jiu-Gang Dong and Li~Qiu.
\newblock Complex {L}aplacians and applications in multi-agent systems.
\newblock {\em arXiv preprint arXiv:1406.1862}, 2014.

\bibitem{DONG20161}
Jiu-Gang Dong and Lin Lin.
\newblock {L}aplacian matrices of general complex weighted directed graphs.
\newblock {\em Linear Algebra and its Applications}, 510:1--9, 2016.

\bibitem{dong2015consensus}
Jiu-Gang Dong and Li~Qiu.
\newblock Consensus problems in complex-weighted networks.
\newblock {\em arXiv preprint arXiv:1406.1862}, 2015.

\bibitem{saxena2024realeventualexponentialpositivity}
Aditi Saxena, Twinkle Tripathy, and Rajasekhar Anguluri.
\newblock Real eventual exponential positivity of complex-valued {L}aplacians: Applications to consensus in multi-agent systems.
\newblock In {\em 2024 Tenth Indian Control Conference (ICC)}, pages 226--231, 2024.

\bibitem{saxena2024flowscomplexvaluedlaplacianspseudoinverses}
Aditi Saxena, Twinkle Tripathy, and Rajasekhar Anguluri.
\newblock Are the flows of complex-valued {L}aplacians and their pseudoinverses related?
\newblock {\em arXiv preprint arXiv:2411.09254}, 2024.

\bibitem{vargamatrix}
Richard~S Varga.
\newblock Iterative analysis.
\newblock {\em New Jersey}, 322, 1962.

\bibitem{varga2012}
Dimitrios Noutsos and Richard~S Varga.
\newblock On the {P}erron--{F}robenius theory for complex matrices.
\newblock {\em Linear algebra and its applications}, 437(4):1071--1088, 2012.

\bibitem{rugh2010cones}
Hans~Henrik Rugh.
\newblock Cones and gauges in complex spaces: Spectral gaps and complex perron-frobenius theory.
\newblock {\em Annals of mathematics}, pages 1707--1752, 2010.

\bibitem{bullo2018lectures}
Francesco Bullo.
\newblock {\em Lectures on network systems}, volume~1.
\newblock CreateSpace, 2018.

\bibitem{Networksgraphtheory2018}
Florian Dörfler, John~W. Simpson-Porco, and Francesco Bullo.
\newblock Electrical networks and algebraic graph theory: Models, properties, and applications.
\newblock {\em Proceedings of the IEEE}, 106(5):977--1005, 2018.

\bibitem{friedland1978inverse}
Shmuel Friedland.
\newblock On an inverse problem for nonnegative and eventually nonnegative matrices.
\newblock {\em Israel Journal of Mathematics}, 29:43--60, 1978.

\bibitem{ballantine1970stabilization}
CS~Ballantine.
\newblock Stabilization by a diagonal matrix.
\newblock {\em Proceedings of the American Mathematical Society}, 25(4):728--734, 1970.

\bibitem{shrinate2024centrality}
Aashi Shrinate and Twinkle Tripathy.
\newblock Absolute centrality in a signed {F}riedkin-{J}ohnsen based model: a graphical characterisation of influence.
\newblock 2024.

\bibitem{veerman2019diffusion}
JJP Veerman and Ewan Kummel.
\newblock Diffusion and consensus on weakly connected directed graphs.
\newblock {\em Linear Algebra and its Applications}, 578:184--206, 2019.

\bibitem{masuda2017random}
Naoki Masuda, Mason~A Porter, and Renaud Lambiotte.
\newblock Random walks and diffusion on networks.
\newblock {\em Physics reports}, 716:1--58, 2017.

\bibitem{masuda2010dynamics}
Naoki Masuda and Hiroshi Kori.
\newblock Dynamics-based centrality for directed networks.
\newblock {\em Physical Review E—Statistical, Nonlinear, and Soft Matter Physics}, 82(5):056107, 2010.

\bibitem{hoser2005eigenspectral}
Bettina Hoser and Andreas Geyer-Schulz.
\newblock Eigenspectral analysis of hermitian adjacency matrices for the analysis of group substructures.
\newblock {\em Journal of Mathematical Sociology}, 29(4):265--294, 2005.

\end{thebibliography}
\bibliographystyle{unsrt}
\vspace{-10mm}
\begin{IEEEbiography}[{\includegraphics[width=1in,height=1.25in,clip,keepaspectratio]{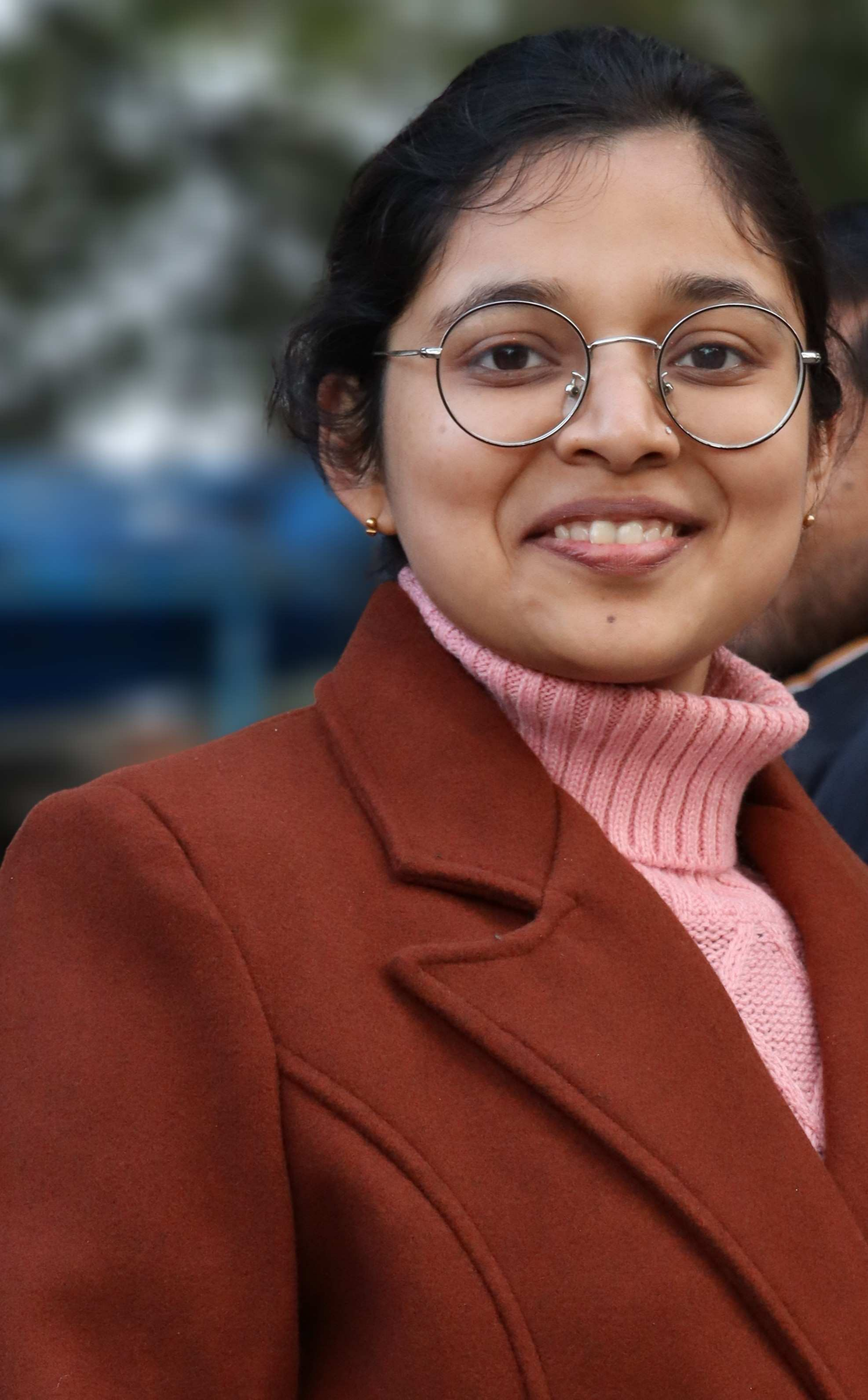}}]{Aditi Saxena} (Student Member, IEEE) recieved her B.Tech degree in Electrical Engineering from Dr. APJ Abdul Kalam Technical University in 2021. Further, she received her M.Tech in Control and Instrumentation from MNNIT Allahabad in 2023. Currently, she is a Ph.D. student in the Department of Electrical Engineering at Indian Institute of Technology Kanpur, India. Her research interests include opinion dynamics in multi-agent systems, network science, control of unmanned aerial vehicles, and nonlinear systems.  
\end{IEEEbiography} \vspace{-10mm}
\begin{IEEEbiography}[{\includegraphics[width=1in,height=1.25in,clip,keepaspectratio]{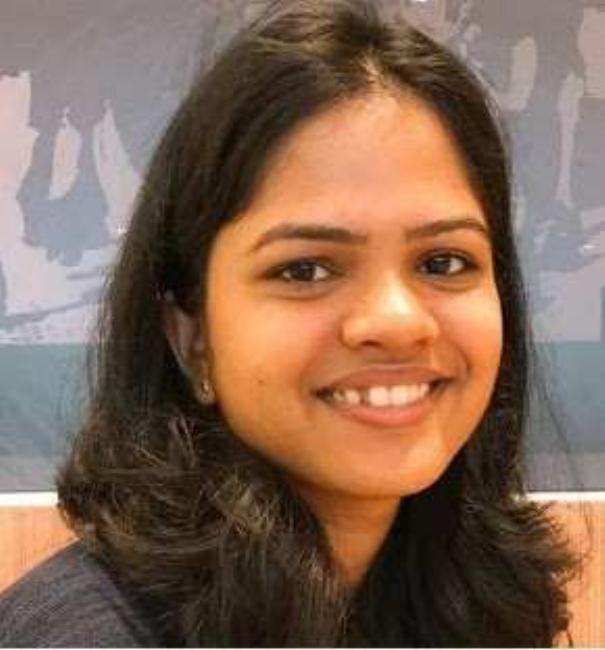}}]{Twinkle Tripathy} 
(Senior Member, IEEE) is currently an Assistant Professor in the Department of Electrical Engineering of IIT Kanpur. She received a Dual Degree of MTech. and Ph.D. at Systems and Control Engineering, IIT Bombay in Dec. 2016. She started her post-doctoral tenure at the School of Electrical and  Electronic Engineering, NTU, Singapore. After serving there for a year, she joined the Faculty of Aerospace Engineering, Technion Israel Institute of Technology as a post-doctoral fellow. Her research interests broadly include control and guidance of autonomous systems, cyclic pursuit strategies and opinion dynamics.
\end{IEEEbiography} \vspace{-20mm}
\begin{IEEEbiography}[{\includegraphics[width=1in,height=1.25in,clip,keepaspectratio]{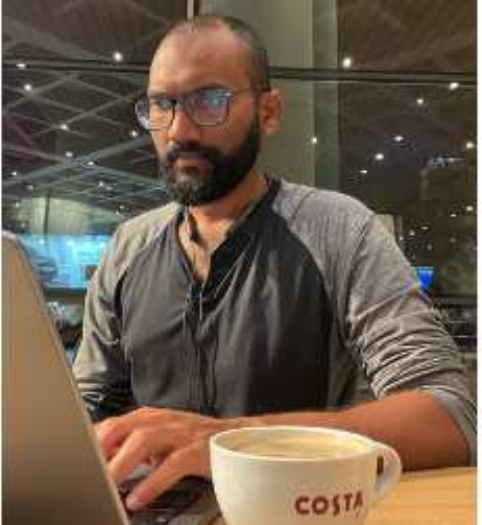}}]{Rajasekhar Anguluri} (Member, IEEE) is an Assistant Professor in the Dept. of Computer Science and Electrical Engineering,  University of Maryland, Baltimore County, MD, USA. He received the B.Tech. degree in electrical engineering from NIT-Warangal, India, in 2013, and the M.S. degree in statistics and the Ph.D. degree in mechanical engineering from the University of California at Riverside, CA, USA, both in 2019.
He was a Postdoctoral Research Scholar with the School of
Electrical, Computer, and Energy Engineering, ASU, Tempe, AZ, USA (20-23). 

His research interests include statistical signal processing, estimation and control, and power
systems. When he is not deeply thinking about matrix mechanics, he thinks deeply on formulating optimization problems on matrices. When he is not thinking deeply on matrices, he can be found in the university library, flipping through old dust jackets. 
\end{IEEEbiography}
\end{document}